\newif\ifcomm
\commfalse

\newif\ifblind
\blindfalse

\newif\ifs
\sfalse

\documentclass[letterpaper,twocolumn,10pt]{article}
\usepackage{usenix-2020-09}

\ifs %
    \usepackage[subtle]{savetrees}
    \usepackage[compact,small]{titlesec}
\else
\fi

\usepackage{amsmath}

\usepackage{tikz}
\usepackage{amsmath, nccmath}
\usepackage[english]{babel}
\usepackage{blindtext}
\usepackage{outlines}
\usepackage{graphicx}
\usepackage{times}
\usepackage{hyperref}
\usepackage{ulem}
\usepackage{xspace}
\usepackage[dvipsnames]{xcolor}
\usepackage{balance}
\usepackage{outlines}
\usepackage{caption}
\usepackage{multirow}
\usepackage[]{url}
\usepackage{todonotes}
\usepackage{algorithm}
\usepackage{algpseudocode}
\usepackage{tabularx}
\usepackage{array}
\usepackage{listings}
\usepackage[font=small,labelfont=bf]{caption}
\usepackage{xurl}
\usepackage[capitalise]{cleveref}
\usepackage[inline,shortlabels]{enumitem}
\setitemize{noitemsep,topsep=0pt,parsep=0pt,partopsep=0pt,leftmargin=1em}
\usepackage{makecell}
\usepackage{multicol}
\usepackage{subcaption}
\usepackage{bm} %
\usepackage{comment}

\usepackage[table]{xcolor} %

\usepackage[font={small,bf},labelfont={small,bf},textfont={small},belowskip=-5pt]{caption}
\usepackage[font={small,bf},labelfont={small,bf},textfont={small},belowskip=-6pt,aboveskip=1pt, labelformat=simple]{subcaption}

\newcommand{\eat}[1]{}
\ifcomm
    \newcommand{\allnotes}[1]{#1} 
\else
    \newcommand{\allnotes}[1]{} 
\fi

\usepackage[english]{babel}
\usepackage{blindtext}

\usepackage{amsmath}
\usepackage{amsthm} %
\usepackage{caption}
\usepackage{graphicx}
\usepackage{xurl}
\usepackage{booktabs}  %
	
\usepackage{array}
	\newcolumntype{C}[1]{>{\centering\let\newline\\\arraybackslash\hspace{0pt}}m{#1}}

\newcommand{\bl}{\begin{enumerate*}[(1)]}
\newcommand{\el}{\end{enumerate*}}

\usepackage{pifont} %

\crefname{equation}{Eq.}{Eqs.}
\Crefname{equation}{Eq.}{Eqs.}
\crefname{figure}{Fig.}{Figs.}
\Crefname{figure}{Fig.}{Figs.}
\crefname{table}{Table}{Tables}
\Crefname{table}{Table}{Tables}
\crefname{property}{Property}{Properties}
\Crefname{property}{Property}{Properties}
\crefname{theorem}{Thm.}{Thms.}
\Crefname{theorem}{Thm.}{Thms.}
\let\oldappendix\appendix
\renewcommand{\appendix}{%
  \oldappendix%
  \crefalias{section}{appendix}%
}

\AtBeginDocument{%
    \crefname{appendix}{}{}%
    \Crefname{appendix}{}{}%
    \crefname{section}{}{}%
    \Crefname{section}{}{}%
    
    \crefformat{appendix}{Appendix~#2#1#3}
    \Crefformat{appendix}{Appendix~#2#1#3}
    \crefrangeformat{appendix}{Appendices~#3#1#4--#5#2#6}
    \Crefrangeformat{appendix}{Appendices~#3#1#4--#5#2#6}
    \crefmultiformat{appendix}{Appendices~#2#1#3}{ and~#2#1#3}{, #2#1#3}{ and~#2#1#3}
    \Crefmultiformat{appendix}{Appendices~#2#1#3}{ and~#2#1#3}{, #2#1#3}{ and~#2#1#3}

    \crefname{subappendix}{}{}%
    \Crefname{subappendix}{}{}%
    \crefformat{subappendix}{Appendix~#2#1#3}
    \Crefformat{subappendix}{Appendix~#2#1#3}
    \crefrangeformat{subappendix}{Appendices~#3#1#4--#5#2#6}
    \Crefrangeformat{subappendix}{Appendices~#3#1#4--#5#2#6}
    \crefmultiformat{subappendix}{Appendices~#2#1#3}{ and~#2#1#3}{, #2#1#3}{ and~#2#1#3}
    \Crefmultiformat{subappendix}{Appendices~#2#1#3}{ and~#2#1#3}{, #2#1#3}{ and~#2#1#3}

    \crefformat{section}{\S#2#1#3}
    \Crefformat{section}{\S#2#1#3}
    \crefrangeformat{section}{\S\S#3#1#4--#5#2#6}
    \Crefrangeformat{section}{\S\S#3#1#4--#5#2#6}
    \crefmultiformat{section}{\S\S#2#1#3}{ and~#2#1#3}{, #2#1#3}{ and~#2#1#3}
    \Crefmultiformat{section}{\S\S#2#1#3}{ and~#2#1#3}{, #2#1#3}{ and~#2#1#3}

    \crefname{subsection}{}{}%
    \Crefname{subsection}{}{}%
    \crefname{subsubsection}{}{}%
    \Crefname{subsubsection}{}{}%
    \crefformat{subsection}{\S#2#1#3}
    \Crefformat{subsection}{\S#2#1#3}
    \crefrangeformat{subsection}{\S\S#3#1#4--#5#2#6}
    \Crefrangeformat{subsection}{\S\S#3#1#4--#5#2#6}
    \crefmultiformat{subsection}{\S\S#2#1#3}{ and~#2#1#3}{, #2#1#3}{ and~#2#1#3}
    \Crefmultiformat{subsection}{\S\S#2#1#3}{ and~#2#1#3}{, #2#1#3}{ and~#2#1#3}
    \crefformat{subsubsection}{\S#2#1#3}
    \Crefformat{subsubsection}{\S#2#1#3}
    \crefrangeformat{subsubsection}{\S\S#3#1#4--#5#2#6}
    \Crefrangeformat{subsubsection}{\S\S#3#1#4--#5#2#6}
    \crefmultiformat{subsubsection}{\S\S#2#1#3}{ and~#2#1#3}{, #2#1#3}{ and~#2#1#3}
    \Crefmultiformat{subsubsection}{\S\S#2#1#3}{ and~#2#1#3}{, #2#1#3}{ and~#2#1#3}
}

\newcommand{\thm}[1]{Thm.~\ref{#1}} %
\newcommand{\fig}[1]{Fig.~\ref{#1}} %
\newcommand{\inv}[1]{Inv.~\ref{#1}} %

\newcommand{\be}{\begin{equation}}
\newcommand{\ee}{\end{equation}}

\providecommand{\vs}{{vs.}\xspace}
\providecommand{\ie}{{i.e.,}\xspace}
\providecommand{\eg}{{e.g.,}\xspace}
\providecommand{\etc}{{etc.}\xspace}

\ifcomm
    \newcommand{\mycomm}[3]{{\footnotesize{{\color{#2} {[#1: #3]}}}}}
    \newcommand{\Fmycomm}[3]{{\color{red} \footnote{{{\color{#2} {[#1: #3]}}} }}}
    \newenvironment{AddBlock}{\color{blue}}{}
\else
    \newcommand{\mycomm}[3]{}
    \newcommand{\Fmycomm}[3]{}

\fi

\newcommand{\I}[1]{}

\newcommand{\T}[1]{\par\smallskip\noindent\textbf{#1}} %
\newcommand{\Ts}[1]{\par\noindent\textit{#1}} %
\newcommand{\mypar}[1]{\T{#1.\xspace}}
\newcommand{\mypars}[1]{\Ts{#1.\xspace}}

\newtheoremstyle{boldthm}{}{}{\itshape}{}{\bfseries}{.}{ }{\thmname{#1}\thmnumber{ #2}\thmnote{ (#3)}} %
\theoremstyle{boldthm}
\newtheorem{invariant}{Invariant}
\newtheorem{theorem}{Theorem}%
\newtheorem{property}{Property} %

\newcommand{\bp}{\begin{proof}}
\newcommand{\bpo}{ \begin{proof}[Proof Outline] }
\newcommand{\ep}{\end{proof}}       %
\newcommand{\proofof}[1]{\begin{proof}[Proof of #1]} %

\newcommand{\para}[1]{\left( #1 \right)}        %
\newcommand{\brac}[1]{\left\{ #1 \right\}}
\newcommand{\sbrac}[1]{\left[ #1 \right]}
\newcommand{\floor}[1]{\left\lfloor #1 \right\rfloor} 
\newcommand{\ceil}[1]{\left\lceil #1 \right\rceil} %

\newcommand{\unit}[1]{\;\mathrm{#1}} %
\usepackage{upgreek} %
\newcommand{\us}{\unit{\mu s}}
\newcommand{\p}[1]{\Pr \para{#1}}  %
\newcommand\bigO[1]{\ensuremath{\mathcal{O}\left(#1\right)}}
\renewcommand{\th}{\ensuremath{^\text{th}}}

\newcommand{\cmark}{\ding{51}}
\newcommand{\xmark}{\ding{55}}
\newcommand{\vx}{\checkmark\kern-1.1ex\raisebox{.7ex}{\rotatebox[origin=c]{125}{--}}} %

\newcommand{\dr}{DR\xspace} 
\newcommand{\DR}{\dr} 

\newcommand{\htsim}{\textit{htsim}\xspace} 
\newcommand{\lblabel}[1]{\textsc{#1}\xspace}
\newcommand{\hpkt}{\lblabel{RPS}} %

\newcommand{\hpktar}{\lblabel{REPS}}
\newcommand{\reps}{\hpktar}
\newcommand{\spktar}{\lblabel{AR}}
\newcommand{\ar}{\spktar}
\newcommand{\jsq}{\lblabel{JSQ}}
\newcommand{\rsq}{\lblabel{RSQ}}
\newcommand{\rps}{\lblabel{RPS}}

\newcommand{\simplerr}{\lblabel{Simple RR}}
\newcommand{\hdr}{\lblabel{DRB}} %
\newcommand{\drb}{\hdr}
\newcommand{\ethereal}{\lblabel{Ethereal}}

\newcommand{\Ofan}{Ofan} %
\newcommand{\name}{\lblabel{\Ofan}} 
 
\providecommand{\sysname}{{\Ofan}\xspace}
 
\newcommand{\pipe}{\lblabel{p\Ofan}} 
\newcommand{\four}{\lblabel{4\Ofan}} 
\newcommand{\one}{\lblabel{1\Ofan}} 

\newcommand{\drill}{\lblabel{DRILL}}
\newcommand{\choice}{\lblabel{Choice}}
\newcommand{\conga}{\lblabel{CONGA}}
\newcommand{\dlb}{\lblabel{DLB}}
\newcommand{\sglb}{\lblabel{SGLB}}

\newcommand{\m}{m}
\newcommand{\n}{n}
\newcommand{\C}{C}
\newcommand{\q}{q(\m)}
\newcommand{\qup}{q_{\uparrow}(\m)}
\newcommand{\qdown}{q_{\downarrow}(\m)}
\renewcommand{\k}{k}
\newcommand{\w}{w} %
\renewcommand{\a}{a}
\newcommand{\pkt}{L_{\text{PKT}}}

\newcommand{\rhomax}{\rho_{\text{max}}}
\renewcommand{\H}{H}
\newcommand{\Td}{T_d}
\newcommand{\Ta}{T_a}
\newcommand{\Tg}{T_g}
\renewcommand{\d}{T'_d}
\renewcommand{\a}{T'_a}
\renewcommand{\p}{T_p}

\begin{document}

\date{}

\title{\Large \bf \Ofan: Optimal Load Balancing for AI Training}

\newcommand{\aut}[2]{#1\texorpdfstring{$^{#2}$}{(#2)}}  %
\author{ \rm
  \aut{Sarah McClure}{1} \quad
  \aut{Evyatar Cohen}{2} \quad
  \aut{Jakob Krebs}{2} \quad
  \aut{Alexander Shpiner}{3} \\
  \rm
  \aut{Mark Silberstein}{2,3}\quad
  \aut{Scott Shenker}{1,4} \quad
  \aut{Sylvia Ratnasamy}{1} \quad
  \aut{Isaac Keslassy}{1,2}
  \\ \\%[0.6em]
  {%
  $^1$ \textit{UC Berkeley}\quad
  $^2$ \textit{Technion} \quad
  $^3$ \textit{NVIDIA} \quad
  $^4$ \textit{ICSI}
  }
}%

\maketitle

\begin{abstract}
The extreme collective completion time (CCT) demands of AI workloads challenge existing packet spraying algorithms, which can have trouble efficiently load-balancing workloads that are sent at full line rates.

We trace this to a structural cause: on a fat tree, once a packet picks its upward path, the downward path to its destination is unique, so destination-oblivious schemes cannot undo the imbalance it creates. We prove that such schemes can suffer from $\Theta(\sqrt{m})$ queueing for messages of size $m$, thus eventually triggering rate reductions by the congestion control. Instead, we suggest \name, a switch-based destination-aware LB scheme that can reach $O(1)$ queueing. We also present its  \pipe variant that fits the pipe-based architecture of current switches. Our P4 implementation shows that it consumes modest resources. An end-to-end FSDP2 evaluation with Llama-3 405B-parameter models shows that \name cuts CCT inflation by $16$--$39\times$ when compared to existing algorithms.
\end{abstract}

\section{Introduction}
\label{sec:intro}

\T{Thirty years of load balancing.}
Datacenter networks offer a rich set of parallel paths, and a long line of load-balancing (LB) algorithms has been proposed to exploit them.
That line was long constrained by the transport protocol's intolerance for reordering, and so operated at a \emph{coarse} granularity: per-flow hashing~\cite{ecmprfc} and its weighted variants~\cite{wcmp,ietf-idr-link-bandwidth-24}, subflows~\cite{mptcp,alibaba,ethereal}, host-based path adaptation~\cite{plb,falcon,meta}, and flowlet switching~\cite{kandulaflowlet,letitflowlet,clove}.
All of these pin a unit of traffic, whether a flow, a subflow, or a flowlet, to a single path for its lifetime.
On the contrary, relying on NICs that accept out-of-order delivery~\cite{nvidia-ooo,uec-spec,conweave}, \textit{spraying} algorithms load-balance at the granularity of a single packet. 
These include (i)~random packet spraying (RPS)~\cite{rps} as in Alibaba's Stellar~\cite{alibaba-stellar}, (ii)~host-based congestion-aware spraying algorithms like REPS and MRC~\cite{strack, bonatoreps, MRC} as in the Ultra Ethernet Consortium~\cite{uec-spec}, and (iii)~switch-based algorithms that steer packets based on local congestion like \drill~\cite{drill} and Nvidia's Spectrum-X~\cite{nvidia_AR_whitepaper,nvidia_AR_whitepaper2,spectrum-x} adaptive routing (AR)~\cite{duato650interconnection,gomez2007deterministic}, or also take into account congestion at neighboring switches~\cite{bytedance-glb,conga}, like Broadcom's GLB and ByteDance's \sglb.

\begin{figure}
    \centering
    \begin{subfigure}[]{1\columnwidth}
    \centering
        \includegraphics[width=0.5\columnwidth]{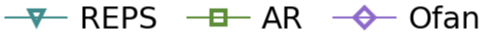}
    \end{subfigure}
    \begin{subfigure}[b]{0.31\columnwidth}
        \includegraphics[width=\columnwidth]{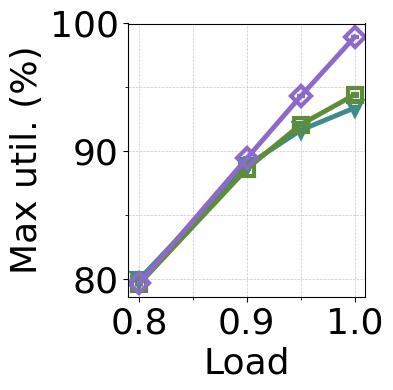}
        \caption{Max utilization}
        \label{fig:load1}
    \end{subfigure}%
    \hspace{0.02\columnwidth}
    \begin{subfigure}[b]{0.31\columnwidth}
        \includegraphics[width=\columnwidth]{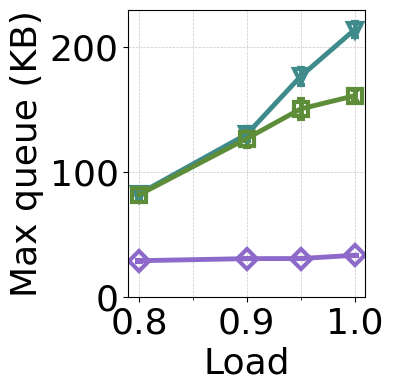}
        \caption{Max queue size}
        \label{fig:load2}
    \end{subfigure}
    \hspace{0.02\columnwidth}
    \begin{subfigure}[b]{0.31\columnwidth}
        \includegraphics[width=\columnwidth]{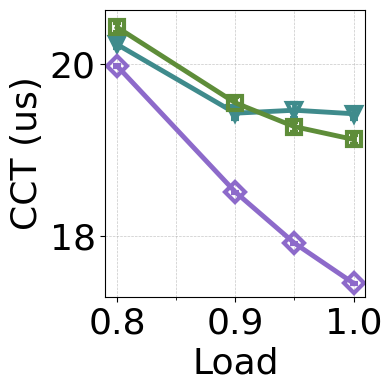}
        \caption{CCT}
        \label{fig:load3}
    \end{subfigure}    
    \caption{Utilization and max queue size as load increases.
    }
    \label{fig:load}
\end{figure}

\T{AI training.}
Recent AI training workloads have new demands.  They focus on minimizing collective completion time (CCT), \ie the worst-case flow completion time among all flows in a collective, and therefore push every source to send at $100\%$ of its line rate~\cite{ncclcollectives,nccl,megascale-moe}.

Unfortunately, as \fig{fig:load} shows, existing algorithms fail to meet this demand. We simulate a three-level fat-tree with a random permutation workload using $1\unit{MB}$ messages (as detailed in \cref{sec:sensitivity}), and sweep the offered load toward $100\%$.
We also remove congestion control entirely, pacing each source at a fixed rate dictated by the load, so that we only measure the effects of load balancing (\S\ref{sec:approach}).
The results display the limits of current state-of-the-art spraying algorithms. \Cref{fig:load1} shows that for REPS and AR, utilization cannot keep up with load beyond load = 0.9 (it disconnects from the $y=x$ diagonal), and at load 1 it saturates below 95\%. Likewise, \Cref{fig:load2} illustrates how their queue sizes of \reps and \ar grow significantly despite the small message size. Finally, \Cref{fig:load3} shows that their CCTs  improve less beyond a load of $0.9$, also pointing to some saturation. 
In this paper, we claim that all these issues come from a sub-optimal load balancing: as packets do not reach all queues equally, some queues become overloaded (high queue sizes) while others remain underutilized (low utilization), thus impacting the CCT.
We introduce a new LB algorithm, \name. These plots show how its efficient LB properties enable \name to achieve about 99\% utilization, much smaller queue sizes, and a clearly better CCT. Note that once congestion control is switched back on, we would expect ECN congestion signals (and PFC pauses for larger messages) to occur for REPS and AR, yielding an even worse CCT gap.

\T{\name.} In the paper, we introduce \name, a new LB algorithm that relies on three key ideas:

\T{(1) Per-destination LB.}
In fat trees, LB has freedom only on the way up, while the path down to a destination is unique.
Existing switch-based LB schemes typically select outgoing port to optimize upward LB. 
This balances the northbound links beautifully. But once a packet has picked its upward path, 
no downstream switch can correct a bad choice.
A locally-efficient but destination-oblivious LB scheme therefore spreads each destination's traffic unevenly when going up, and that unevenness cannot be addressed when going down. 

In this paper, we consider random permutation workloads with message size $\m$ and assume that all hosts blast traffic at 100\% load. We prove that  randomized spraying algorithms incur an average queue size of $\Omega(\sqrt{\m})$ (\thm{thm:sqrt}). We also show that in evaluations, additional destination-oblivious algorithms follow a similar heuristic growth. 
In practice, these bounds  mean that as message sizes grow, queues grow, and therefore congestion control algorithms (CCAs) will be forced by ECNs and PFCs to slow down sending rates and increase CCTs, hurting AI training performance.

Instead, inspired by host-based DRB~\cite{drb} and \ethereal~\cite{ethereal}, we propose a \emph{destination-based rotation} (\DR) discipline, in which consecutive packets to a given destination are deterministically rotated across paths. Using  $N\,D/D/1$ queueing models~\cite{menth2010packet, roberts1996broadband}, we show that destination-based rotation can achieve optimal $\bigO{1}$ queue bounds even at 100\% loads (\thm{thm:DRB}).

\T{(2) Switch-based LB.}
We argue that the switches are a better home for \DR than the hosts.
First, the switch already knows the full topology and its link failures using BGP, while the host must either implement BGP in the NIC or keep re-inferring it indirectly. 
Second, and this is the property we find most surprising, a switch can \emph{consolidate} traffic for a given destination, making packets from different sources share the same rotation pointer.
A host, seeing only its own flows, cannot do this.

\T{(3) Pipe-based LB.} 
Finally, we show that we can exploit current \textit{pipe}-based switch implementations (called pipes in Broadcom's Tomahawk switches and slices in Cisco’s Silicon One switches). Pipes aggregate a few ports (\eg 4 ports at 800Gb/s each) and guarantee atomic Read-Modify-Write (RMW) accesses to local registers at line rate. Thus, we introduce \pipe, a more practical variant of \name that stores a local pointer at each pipe for each destination instead of a shared pointer accessible across the whole switch. 

\T{Contributions.} In summary, we make the following contributions:
\begin{itemize}
    \item We demonstrate that existing randomized-spraying algorithms suffer from an intrinsic queuing lower bound that grows with the message size $m$ as $\Omega(\sqrt{\m})$. We also demonstrate that some LB algorithms like join-the-shortest-queue (JSQ) even suffer from newly discovered synchronization issues at full line rate, and suffer from a $\Theta(m)$ queue growth. On the contrary, we show that destination-based rotation achieves $\bigO{1}$ queue bounds that enable larger messages without queue growth (\cref{sec:optimality}).
    \item We present \name, a switch-based \DR, and explain the intuition behind its consolidation properties (\S\ref{sec:algo}).
    \item We further present how to implement its variant \pipe in a pipe-based switch, and how it requires limited  resources in a P4 implementation for an Intel Tofino switch (\cref{sec:implementation}).
    \item Finally, we show how network size, message size, packet size, collective type,  link failures, topology, incremental deployment and congestion control affect the results. We find that \name and \pipe keep outperforming existing algorithms across a wide variety of conditions. 
    We further evaluate an AI training scenario with Llama-3 (405B parameters) using FSDP2, and rely on two different recent CCAs designed for spraying: UEC's NSCC~\cite{uec,uec-spec} and MSwift~\cite{mswift}. \name reduces the CCT increase beyond a CCT lower-bound by at least $16\times$ with NSCC and $39\times$ with MSwift when compared to AR and REPS  (and \pipe with four ports per pipe reduces by $15\times$ and $36\times$, respectively). (\S\ref{sec:sensitivity}).
\end{itemize}

The rest of the paper is organized as follows:
\S\ref{sec:background} states the workload properties our analysis assumes and the methodology that isolates the load-balancing layer;
\S\ref{sec:optimality} explain the intuition behind per-destination LB, proves the limits of most existing algorithms, and offers \DR as a way forward;
\S\ref{sec:algo} presents \name and \S\ref{sec:implementation} its realization in switch hardware;
\S\ref{sec:sensitivity} evaluates it, and \S\ref{sec:discussion} discusses what remains open.%
\ifblind
\else
\footnote{This paper is a follow-up on \cite{LB-arxiv}, which had two parts: (1)~systematically evaluating leading load-balancing designs, while decoupling them from specific congestion control and loss recovery stacks; and (2)~introducing and studying \name. This paper focuses on the second part.}
\fi

\section{Background and Methodology}
\label{sec:background}

\subsection{AI Training Workloads}

AI training alternates rounds of computation with rounds of communication, where the communication is one of a few standard collectives, such as AllReduce, AllGather, and AlltoAll, drawn from a small set of libraries~\cite{ncclcollectives,nccl,parallelism}.
Two properties of this traffic matter for load balancing:

\mypar{Uniform} The nodes in a collective run the same communication library and the same transport and NICs, and most collectives are symmetric: every participant has the same amount of data to send, with the same packet size. 

\mypar{Synchronous} All nodes begin at roughly the same time~\cite{ethereal}. This cuts both ways. It helps load balancing, since there are no newly arriving flows to adapt to mid-collective. But it also means that every source steps through its packets in lockstep, which can cause a load balancer's port choices to synchronize across sources and degrade catastrophically (as shown in \S\ref{sec:optimality}).

Together with the last-byte nature of CCT, these properties also mean that the traffic is long-lived and gapless: techniques that rebalance during idle periods, such as PLB or flowlet switching, are less helpful. 
We also assume the cluster runs the job in isolation, as is common at scale~\cite{meta,alibaba,openai-ms-blog}.

Throughout this paper, we will focus on three types of generic traffic matrices that reflect generic AI training workloads~\cite{meta100k}: (i)~\textit{random permutations}, (ii)~\textit{all-to-all}, and (iii)~\textit{MoE (all-to-all-v)}. In particular, we will show that  all-to-all is the easier one to load-balance, and all-to-all-v suffers from a large incast that reduces the impact of load-balancing. Therefore, we will most focus on permutations, where the LB scheme can make the biggest impact. 

\subsection{Isolating the Load-Balancing Layer}
\label{sec:approach}
Load balancing does not run alone: congestion control sets sending rates, and loss recovery decides what to retransmit. Both react to the queues that load balancing creates, so measuring a full stack conflates the three. Our methodology in the analysis is therefore to remove the other two and measure load balancing on its own\ifblind{.}\else \xspace(as in~\cite{LB-arxiv}). \fi In evaluations (\cref{sec:sensitivity}), we add  CCAs back, and find that they do not change the main results, including the ranking among LB algorithms.

\mypar{Congestion control} 
A perfect load balancer makes the network's entire capacity available, and an ideal congestion control algorithm (CCA) would provide each host with its fair share of that capacity. 
Under uniform AI collectives, the fair share is the same for all hosts. Thus, we cancel the CCA and allow instead each host to send at the maximum fraction of the line rate that the network can support, which we call $\rhomax$. 
In a fault-free fabric, each host can blast its traffic at full line rate. 
Under failures, $\rhomax$ is the highest rate at which all flows could send if traffic were balanced perfectly across all valid shortest paths. In addition, if a host has several flows, it simply performs round-robin (RR) to schedule them.
Any queue that then forms, and any capacity left unused, is attributable to the load balancer and nothing else. This is what \fig{fig:load} measures.

\mypar{Loss recovery} Similarly, we attempt to avoid any losses: we assume large buffers and do not consider failed links (unless mentioned).

This idealization is deliberate, and it is what lets us state an optimality criterion at all: a bound proved against a fixed-rate source is a property of the load balancer, not of a particular CCA's reaction to it. It is also what makes the criterion demanding, since there is no rate reduction available to rescue a scheme whose queues grow.

\section{Per-Destination LB}
\label{sec:optimality}

In this section, we present theoretical analysis that gives us a more rigorous understanding of optimality in load balancing approaches, and leads us to later propose \sysname, a new and optimal approach, in \S\ref{sec:algo}.

\begin{table}
\caption{Considered LB schemes.
}
\label{tab:models}
\centering
\setlength{\tabcolsep}{2pt} %
\resizebox{\linewidth}{!}{%
\scriptsize
    \begin{tabular}{l l l c l l l}
        \toprule
        \textbf{Scheme} & \textbf{Where} & \textbf{Unit} & \makecell{\textbf{Dest.-}\\\textbf{based}} & \makecell{$\qup$\\\textbf{(model)}} & \makecell{$\qdown$\\\textbf{(model)}} & \textbf{Assumptions} \\
        \midrule
        \multicolumn{7}{l}{\emph{Destination-oblivious LB}}\\
        \simplerr & switch & packet & \xmark & $\Theta(1)$ & $\Theta \para{\m}$ & --- \\
        \jsq~\cite{duato650interconnection} & switch & packet & \xmark & $\Theta(1)$ & $\Theta \para{\m}$ & --- \\
        \conga~\cite{conga} & switch & packet & \xmark & --- & --- & Pkt-based \\
        \choice~\cite{drill} & switch & packet & \xmark & --- & --- & --- \\
        \drill~\cite{drill} & switch & packet & \xmark & --- & --- & Pkt-based \\
        \dlb~\cite{broadcom-dlb} & switch & packet & \xmark & --- & --- & Local \sglb \\
        \sglb~\cite{bytedance-glb} & switch & packet & \xmark & --- & --- & --- \\
        \rps~\cite{rps} & host & packet & \xmark & $\Theta \para{\sqrt{\m}}$ & $\Theta \para{\sqrt{\m}}$ & --- \\
        \rsq & switch & packet & \xmark & $\Theta \para{\sqrt{\m}}$ & $\Theta \para{\sqrt{\m}}$ & --- \\
        \reps~\cite{bonatoreps} & host & packet & \xmark & --- & --- & ECN thres.: 10\% \\
        \ar~\cite{nvidia_AR_whitepaper} & switch & packet & \xmark & --- & --- & Quanta: 5-10-20\% \\

        \midrule
        \multicolumn{7}{l}{\emph{Destination-based LB}}\\
        \hdr~\cite{drb} & host & packet & \cmark & $\Theta(1)$ & $\Theta(1)$ & --- \\
        \ethereal~\cite{ethereal} & host & subflow & \cmark & $\Theta(1)$ & {$\Theta(1)$} & --- \\
        \four (\S\ref{sec:implementation}) & switch & packet & \cmark & $\Theta(1)$ & {$\Theta(1)$} & --- \\
        \name (\S\ref{sec:algo}) & switch & packet & \cmark & $\Theta(1)$ & $\Theta(1)$ & --- \\
        \bottomrule
    \end{tabular}%
} %
\end{table}

\subsection{LB Algorithms}
\label{sec:space}

\Cref{tab:models} lists load-balancing (LB) schemes that we consider in this section. They are all packet-based spraying algorithms except for \ethereal, which relies on subflows, and can be either host- or switch-based. It also lists some assumptions we made in simulations: \eg the initial \conga fits flowlets, and we adapt it to packet spraying to fit our workloads. Likewise, Broadcom's \dlb is not described in detail, but the \sglb paper~\cite{bytedance-glb} sketches it as an \sglb-like algorithm that only uses the 2 local types of information out of the 5 described ones, so we implement it in this way. While all schemes take into account the destination to decide at each switch whether to send packets up or down, the \textit{destination-oblivious LB} schemes consider all packets that go up in the same way. 

We also introduce a few fundamental LB algorithms in the table, as they help convey intuition: \textit{\simplerr} consists of  a simple round-robin model in which the switch load-balances all northbound packets across a set of next-hop ports in round-robin order. In \textit{Join-the-Shortest-Queue (JSQ)}, each switch selects the next-hop port with the shortest queue for a northbound packet, with random tie-breaking between equal-length queues. It is a non-quantized version of \ar. In addition, in \textit{Randomized Switch Queueing (\rsq)}, each switch sends each northbound packet to a random port, and is analytically equivalent to host-based \rps, which changes the packet entropy (\eg source port) to modify the ECMP header-based hashed port decision at each switch.

\begin{figure*}
    \centering
    \begin{subfigure}[b]{0.53\linewidth}
        \includegraphics[width=\textwidth]{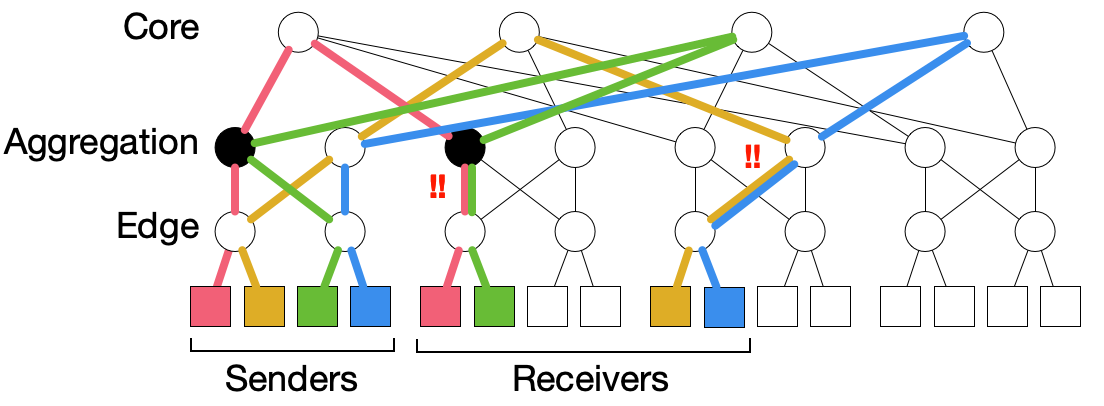}
        \caption{Destination-oblivious LB}
        \label{fig:cong1} %
    \end{subfigure}%
    \hfill
    \begin{subfigure}[b]{0.45\linewidth}
        \includegraphics[width=\textwidth]{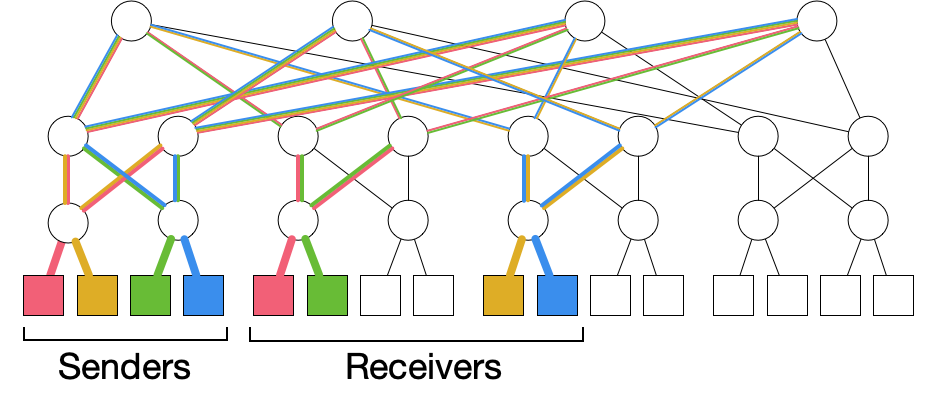}
        \caption{Destination-based LB}
        \label{fig:cong2}
    \end{subfigure}%
    \caption{Impact of  destination-based LB. The three-level fat tree with $\k=4$ includes 16 hosts, four pods (each with two edge and two aggregation switches), and  four core switches. (a)~Destination-oblivious LB can result in conflicts on the way down (red ``!!''), while (b)~destination-based LB avoids them.
    }
    \label{fig:cong}
\end{figure*}

\subsection{The Downlink Problem}
\label{sec:downlink}

Throughout the paper, we consider three-level fat-tree topologies. In such topologies, LB flexibility exists strictly in upward paths, while downward paths to any destination are fixed. 
Consequently, while existing  destination-oblivious LB algorithms may perfectly optimize for local LB properties on the way up, they may also distribute traffic bound for the same edge or aggregation switch unevenly across upward paths, thus creating downstream bottlenecks.

\Cref{fig:cong1} illustrates how it is not sufficient for switches to simply balance traffic on the up routes. Each switch in the senders pod in the figure nicely sends half of its ingress traffic on each link, thus all up links are equally loaded. However, in this extreme example, all traffic to any destination edge switch ends up on the same aggregation switch in the senders pod. While this is fine for the up routes, this can cause collisions on the down routes, where destinations determine the route, not the LB. For example, the red and green paths meet at the first aggregation switch of the first pod, and therefore have to meet again at the first aggregation switch of the second pod (both drawn in black), and therefore conflict on the way down.
In \cref{fig:cong2}, however, switches not only balance load equally across all uplinks, they also ensure that the traffic sent to a particular destination is spread evenly. This ensures that there are no collisions on the down routes as an equal share of each flow traverses every possible path.
Thus, this example shows how a destination-oblivious scheme can balance one side of the network and leave the other to chance, effectively causing statistical imbalance.

\begin{figure}
    \centering
    \begin{subfigure}[]{0.65\columnwidth}
        \centering
        \includegraphics[width=\linewidth]{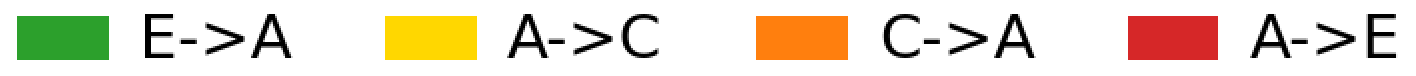}
    \end{subfigure}%
    \par 
    \begin{subfigure}[t]{\columnwidth}
        \centering
        \includegraphics[width=0.9\linewidth]{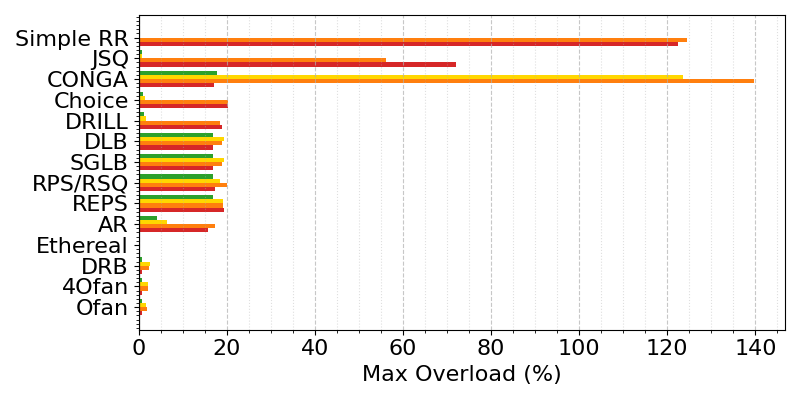}%
    \end{subfigure}
    \caption{Worst-case link overload at each stage. %
    }
    \label{fig:overload}
    \vspace{-0.5em}
\end{figure}

\fig{fig:overload} quantifies the imbalance. We plot the worst-case link overload beyond the ideal, uniform amount, at each of the four inter-switch stages: edge to aggregation ($E\to A$), aggregation to core ($A\to C$), core to aggregation ($C\to A$), and aggregation to edge ($A\to E$). We assume a random inter-pod permutation that forces all flows to cross pods and in which all links are ideally equally loaded. 
The asymmetry is stark. For example, \simplerr and \jsq achieve essentially perfect balance on the two \emph{uplink} stages and then collapse on the two \emph{downlink} stages; \choice and \drill also exhibit similar profiles; the other destination-oblivious schemes suffer at both uplink and downlink. Only the destination-aware schemes are near-flat across all four, as we will detail below. In particular, \ethereal achieves ideal balance because it can decompose a single message into many subflows with one small packet in each.

\subsection{Queue-Size Lower Bounds}
\label{sec:theory}

We now want to provide a fundamental analysis of the performance of each LB scheme. We consider $n$ hosts in a three-level $\k$-ary fat tree without link failures. The hosts perform a random permutation collective, such that each host sends a single flow to a distinct random host (random derangement).  All hosts send traffic at line rate, buffers are infinite, and there are no packet losses nor ACKs. 

We define optimality in terms of the average queue length that a scheme produces.
Collectives can have a large range of message sizes~\cite{meta100k}, and we want to ensure that as they get larger, the queue size does not grow without bound.
For a message size $\m$, let $\q$ be the average queue size across \emph{all} switch queues, over the time it takes for \emph{all} hosts to finish sending. An \emph{optimal} load-balancing algorithm is one that achieves $\q = \bigO{1}$, independent of $\m$. We also decompose $\q=\qup+\qdown$ to obtain the contribution of the uplink and downlink queues, respectively.

Queue size is the right target for several reasons. First, it is strongly correlated with CCT performance: given fixed propagation delay and our assumptions of uniform traffic, ideal sending rates, and zero loss, queueing is the dominant variable component of the CCT, as we validate in \S\ref{sec:theory-sims}. 
Second, it is also a dominant factor in congestion control, as it incurs ECNs, PFCs and buffer overflows. 
Third, it directly affects host NIC reassembly buffer size and complexity.
Finally, an $\bigO{1}$ queue is valuable in itself, since it makes performance predictable.

\mypar{Square-root growth}
We find that in algorithms based on random spraying, both uplink and downlink queues grow as the square root of the message size. The proof, presented in \cref{sec:sqrt}, is based on the fact that the number of packets that land on a given uplink or downlink is a sum of independent random variables, each variable corresponding to whether a packet lands on that link. It finds that the sum has a standard deviation of $\sqrt\m$, which translates to a queue growth of $\sqrt\m$. For downlinks, the proof relies on an approximation that the variations in uplink queue delays are negligible. Later, we show that simulations confirm the square-root growth (\cref{fig:max}).

\begin{theorem}
\label{thm:sqrt}
\rsq and \hpkt, which are based on random spraying, see both their uplink and downlink queues grow as the square root of the message size: $\qup=\Theta(\sqrt\m)$ and $\qdown=\Theta(\sqrt\m)$ (neglecting upstream queueing delays). Hence
\begin{equation} \label{eq:sqrt}
    \q = \Theta \para{\sqrt{\m}}
\end{equation}
\end{theorem}

As we saw in \cref{fig:overload}, some destination-oblivious schemes, such as \drill and \ar,
perform well locally on the uplinks, while others perform more poorly. However, all destination-oblivious schemes suffer from the downlink problem. Intuitively, it is because they simply do not optimize for this downlink and pick uplink ports independently of the exact destination of the packets. We later show that they all exhibit queues that grow with the message size, often in what appears as the square root of $m$ (\cref{fig:max}).

\mypar{Linear growth}
While we proved above a square-root queue growth with the message size, we can also prove a worse queue growth for two more destination-oblivious algorithms (proof in \cref{sec:sync}):

\begin{theorem}
\label{thm:sync}
\simplerr and \jsq keep their uplink queues bounded, $\qup=\Theta(1)$, but suffer from a linear growth of their downlink queues, $\qdown=\Theta(\m)$. Hence
\begin{equation} \label{eq:linear}
    \q = \Theta \para{\m}
\end{equation}
\end{theorem}

\thm{thm:sync} is the more surprising result, especially for \jsq, which is often presented as near-optimal~\cite{duato650interconnection,nvidia_AR_whitepaper,nvidia_AR_whitepaper2,gomez2007deterministic}. The cause is \emph{synchronization}, and it is a direct consequence of the workload uniformity described in \S\ref{sec:background}: all senders transmit the same message and packet size, paced at identical rates. This uniformity creates per-sender ``stickiness,'' in which successive packets from the same sender are forwarded out the same port to the same aggregation switch. 
For example, consider again an edge switch whose $\w=\k/2$ flows all go northbound through its $\w$ ports, and assume \simplerr. Since all flows are paced at the same rate, their packets arrive in the same order every period. So if the pointer sends a packet of some flow to port $i$, the pointer advances exactly $\w$ times before that flow's next packet, and returns to port $i$: the flow ends up sticking to port $i$ for the whole run.

Given the fat tree, such sticky packets necessarily arrive at the same aggregation switch in the destination pod, then stick to the same path all the way down---where they collide with other sticky flows. Note that local \emph{salting} does not help: even if each switch starts \simplerr at a random port and proceeds in a different periodic order, the same synchronization arises. We discuss this in detail in \cref{sec:sync}.

\subsection{Destination-Based Rotation}
\label{sec:dest-aware}

Given these results, we looked to the broader scheduling literature for approaches that reach $\bigO{1}$. We find inspiration in the DRB scheduler of Cao \textit{et al.}~\cite{drb}, whose key idea is that each source host load-balances packets in round-robin fashion on a \emph{per-destination} basis. We generalize this to a scheduling discipline we call \emph{Destination-Based Rotation} (\DR).
We later contribute a switch-based \DR scheme in \S\ref{sec:algo}.

\mypar{\drb as host \DR} \hdr maintains round-robin pointers at each source host, one per (destination host, quantized packet size) pair, using two quanta to separate data and ACK packets. For each (host, size) pair it finds the lowest common switch layer in the topology and 
round-robins among the switches at that layer with a single pointer. For example, each colored source host in \cref{fig:cong} sends its data packets to the four cores in round-robin order; this is done using encapsulation. Under failures, \hdr assumes hosts have full topology knowledge and round-robins only over switches reachable by both endpoints.

\mypar{Bounded queues} Between source and destination, \hdr flows are load-balanced equally and periodically across all allowed paths, so any given path sees a sum of many periodic flows with random phases. We can therefore model \hdr's queues with the $N\,D/D/1$ model~\cite{menth2010packet,roberts1996broadband}, which is particularly noteworthy for \emph{its bounded average queue size even when the arrival rate equals the service rate} (this does not happen in M/M/1 for example). {We prove} the following result in \cref{sec:DRB}. Note that the proof here is weaker than for \cref{thm:sqrt,thm:sync}, as we prove it for the \drb $N\,D/D/1$ approximation, not for the exact algorithm.

\begin{theorem}
\label{thm:DRB}
The average queue size for the \textit{\hdr} model is bounded:
\begin{equation} \label{eq:DRB}
    \q = \Theta \para{1}.
\end{equation}
\end{theorem}

\mypar{\ethereal} A second destination-aware design, 
 \ethereal~\cite{ethereal}, splits each flow into many subflows. Specifically, under 
 a permutation, each host has a single flow, which \ethereal splits into as many equal-size subflows as there are paths, \eg one subflow per core switch for an inter-pod destination. Since the NIC serves its queue pairs in round-robin, the host effectively rotates its packets across all these paths, much like \drb. Furthermore, 
 ACKs follow the reverse path of their data packets, so they are also equally load-balanced across all paths. 
 
 Thus, while \ethereal{} is not a \DR scheme, since it keeps no rotation pointers, its per-destination subflows achieve a similar effect, and we obtain the same bounded-queue result (proof in \cref{sec:ethereal}):
\begin{theorem}
\label{thm:ethereal}
The average queue size for the \ethereal model is bounded:
\begin{equation} \label{eq:ethereal}
    \q = \Theta \para{1}.
\end{equation}
\end{theorem}

\mypar{Deployment challenges} Both schemes ask for something hosts do not have today, exhibiting significant deployment challenges.
\hdr requires each host to know the complete set of currently reachable switches, which in practice means running the fabric's routing protocol (\eg BGP) on the NIC. Exposing routing and link-state updates to hosts violates the layering that separates them from the fabric's internals, and requires every NIC to maintain global topology state and re-converge on every failure.
\ethereal requires the NIC to maintain a very large number of queue pairs (QPs) per flow (one per core, \ie on the order of $1$K with 64-port switches) in order to place each subflow on its own path. Increasing the number of QPs by 3 orders of magnitude could prove daunting at scale.
Both limitations are consequences of running per-destination load balancing \emph{at the host}. A host sees only its own flows and has no native view of the fabric. As we show next, a switch sees both.

\subsection{Validation via Simulation}
\label{sec:theory-sims}

Since our theoretical analysis makes a number of simplifying assumptions, we present simulation results showing both the CCT and the max queue size for the various algorithms discussed above. For the CCT metric, to focus on the variable part impacted by load-balancing, we compute the CCT increase beyond a lower-bound baseline that we derive for the CCT of a permutation (\cref{sec:lowerbound}), and use naive lower bounds for non-permutation workloads. 
For completeness, we also show the performance of \pipe and \name, which are discussed later. 

\begin{figure*}
    \centering
    \begin{subfigure}[]{\linewidth} %
        \centering
        \includegraphics[width=0.7\linewidth]{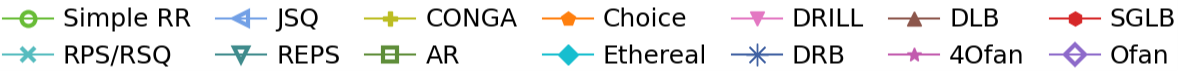}
    \end{subfigure}  
   \par  
    \begin{subfigure}[t]{0.47\linewidth}
        \centering
        \includegraphics[width=0.65\linewidth]{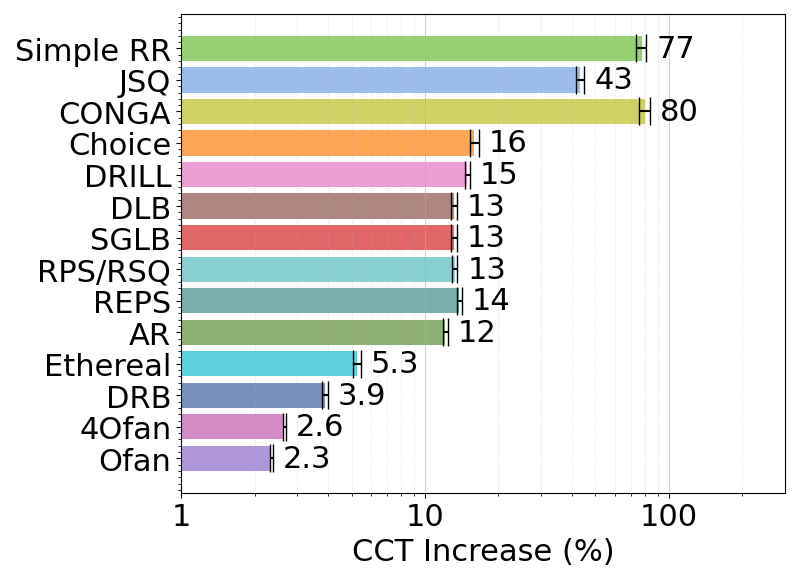} %
        \caption{CCT increase }
        \label{fig:spray-cct}
    \end{subfigure}
    \hfill
    \begin{subfigure}[t]{0.51\linewidth}
        \centering
        \includegraphics[width=0.65\linewidth]{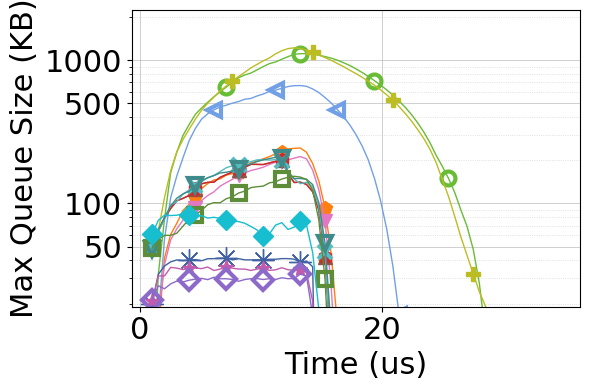}
        \caption{Max queue size}
        \label{fig:max}
    \end{subfigure}
    \caption{Comparison of packet-based LB algorithms.
    }
    \label{fig:growth}
\end{figure*}

\Cref{fig:growth} shows that max queue size correlates well with the CCT metric from our earlier sections. We also see the general trends revealed by our theoretical analysis. Specifically, algorithms cluster in three groups: (1) \simplerr and \jsq (together with \conga) form a first group with near-linear growth of the max queue size; (2) The second group includes \hpkt, \rsq, \hpktar, and \spktar whose queue growth is closer to the square-root of the message size, and (3) \ethereal, \drb, \four and \name  appear to reach an optimal $O(1)$ queue size.  

Given these results, we next take a closer look at practically realizing \DR-based LB.
Since \drb and \ethereal face significant deployment challenges, and \name outperforms them in our simulations, we focus on switch-based \ar and host-based \reps as leading alternatives in the remainder of the evaluation (except for \cref{fig:vsDRB}).

\section{\name} \label{sec:algo}

\subsection{Algorithm}

\S\ref{sec:dest-aware} left us with a discipline that is optimal but has only ever been realized at the host, where it is impractical. This section makes the case that the switch, rather than  the host, is the natural home for \DR, and presents \name, a switch-based realization of it.

\mypar{Switch benefits} We find that the switch enjoys two significant benefits when implementing \DR. 
First, {the switch already knows the topology}. It knows through BGP which links are currently alive. \drb needs the host to obtain this, \eg by implementing somehow BGP onto the NIC; \ethereal keeps iterating to guess it from the host's viewpoint.
A switch gets it for free, and re-converges quickly on failure through the mechanisms it already runs.

Second, and less obviously, the switch can see more traffic. This enables it to  apply round-robin algorithms more efficiently, as described below. 

\mypar{Overview of \name} 
\textit{\name}\footnote{\name's name is based on Ezekiel's vision of wheels within wheels, fitting the many pointer rotations.} 
is a switch-based realization of \DR scheduling. Its goal is to maintain per-destination pointers at edge and aggregation switches, and use these pointers to select the outgoing ports for northbound packets. However, there are many possible destination hosts, leading to an impractical number of pointers to implement.

To address this scalability challenge, \name exploits the structured nature of a fat-tree topology to \emph{consolidate} destinations based on their location in the tree, such that a switch only maintains a single pointer for groups of destinations that share a common edge or aggregation switch. The idea is that a switch does not need to care about the exact destination of its northbound packets, but can consolidate these packets based on their common last-mile switches. 

Specifically, an edge switch $E$ will maintain one pointer for each group of  destinations that share an edge switch $F$ ($F \ne E$). Similarly, an aggregation switch $A$ maintains one pointer for each group of destinations that share an aggregation switch $B$ ($B \ne A$). Hence, an edge switch on the northbound path will round-robin all traffic bound for the same destination edge switch, and an aggregation switch will round-robin traffic bound for the same destination aggregation switch. 

\name's consolidation of destinations significantly reduces the number of pointers, while retaining \DR's property of balancing traffic on both uplinks and downlinks. To see this, consider two flows that originate at source hosts under source edge switch $E$, and are destined to destination hosts under destination edge switch $F$. Then $E$ and $F$ are \textit{mandatory waypoints} in their path. Therefore, $E$ can maintain a single pointer that consolidates both flows. Consolidation is even more effective at the aggregation layer due to the properties of three-level fat-trees. In a fat-tree, all packets destined for a pod and come through a source aggregation switch in their uplink path can only go through a single aggregation switch in their downlink path.\footnote{To see this, notice that any packets leaving the first black aggregation switch in \cref{fig:cong} can only traverse the second black aggregation switch on their southbound path to pod 2.} This destination aggregation switch is again a mandatory waypoint, hence a source aggregation switch can consolidate its scheduling of all packets destined to this pod. %

\mypar{Per-destination pointers} The core idea in \name is to have \emph{a pointer per destination}, and each incoming packet that goes up in the fat tree accesses its specific pointer to determine its northbound port.
For the consolidation explained above, we want to have  at each source edge switch a different pointer per destination edge switch, and at each source aggregation switch a different pointer per destination pod.
Assume that each such switch has $\w=\k/2$ northbound ports. At start up, each pointer is initialized to a random value in $\brac{0,\dots,\w-1}$. 
Then, at each packet arrival for this pointer, it is incremented (modulo the number $\w$ of uplink ports). 

\mypar{Random port order} A simple way is for each packet to be sent to the port indicated by the pointer. For example, if the pointer for the packet's destination indicates 3, then it will be sent to port 3. However, to avoid a possible synchronization across pointers, we don't want all of them to go RR in the same order. Instead, we create a random port traversal order by mapping each pointer value to an uplink port through a pointer-dependent permutation.
That is, even though each pointer goes cyclically through the values $\brac{0,1,2,\dots,\w-1}$ in the same RR order, its permutation will then respectively map them to some port order like $\brac{1,\w-1,2,\dots,0}$. The easiest way to implement these random per-pointer permutations would be to store all of them. However, this would need a large storage space. Instead, we use a simple formula that takes the pointer value and pointer ID and returns the corresponding port ID using a XOR and an addition modulo $\w$  (see \cref{sec:ofan-orders} for details).\footnote{As an order of magnitude, the randomized port order reduces CCT by $\approx 0.2\%$ for permutation workload, but increases it by $\approx 0.1\%$ for all-to-all; so its main appeal is in avoiding or reducing unforeseen synchronization, not necessarily in the immediate CCT reduction.} 

\mypar{ACKs} In addition, using the same pointer for data and ACK packets would not load-balance the load equally, which is our goal: \eg we could send a data packet on one port, then an ACK packet on the next port, \etc Thus, we follow \drb and distinguish pointers for different packet sizes. Specifically, in \name, an edge switch maintains a pointer for each (destination edge switch, quantized packet size) pair and an aggregation switch maintains the same for each (destination pod, quantized packet size). In our simulations, we use two quantization groups for packet sizes: one for data packets, with fixed frame sizes for data payloads, and another for ACKs. Other datacenters and workloads might require more fine-grained quantization. We initialize data pointers at random, and ACK pointers half a rotation away (offset $\floor{\w/2}$).

\mypar{Memory needs} Thus, an edge switch maintains O(\# edge switches) pointers, and an aggregation switch maintains O(\# pods) pointers. 
For example, in a fat-tree of parameter $\k=64$ with $65K$ GPUs~\cite{metaDC,nvidiaDC,broadcomSwitch,nvidiaSwitch}, each edge switch would maintain $4K$ pointers of $\log_2{\w}=5$ bits each, totalling $2.5\unit{KB}$, and each aggregation switch $126$ pointers of 5 bits each, \ie $79\unit{B}$. We add $20\unit{B}$ to store the same basis permutation in each switch.  
Our discussions with switch vendors suggest that this is a very reasonable number of pointers and storage amount for switches to maintain.

Note that unlike \drb, \name does not require any tunneling or encapsulation of traffic, nor any change at endhosts. Of course, it does require changes at edge and aggregation switches.

\mypar{Failures} When there are link failures in the topology, the pointers still cycle through the ports in the same way. However, if they pick an invalid port for a packet, \ie a port through which there is no possible path to the destination, then they just keep incrementing until yielding a valid port. \footnote{We assume that the operator has a procedure in place in case the switch is a dead end.}

\subsection{Bounded queue sizes}

\mypar{Properties} 
We get the same $\Theta \para{1}$ result for \name's model as for \drb (again, we only prove it for the queueing model, see \cref{sec:ofan}):
\begin{theorem}
\label{thm:ofan}
The modeled average queue size for \textit{\name} is bounded:
\begin{equation} \label{eq:ofan}
    \q = \Theta \para{1}.
\end{equation}
\end{theorem}

\mypar{Consolidation pays off} \Cref{fig:overload} in \S\ref{sec:downlink} already showed \name flat across all four link stages. It also showed the effect consolidation buys over \hdr: \hdr is perfectly balanced for large messages but degrades for small ones, because its many per-flow pointers spread traffic too thin, leaving a constant residual imbalance that a small message cannot amortize. \name's far smaller number of consolidated pointers avoids this. In \cref{sec:sensitivity} we show that the gap widens with network size, tracking the widening disparity in pointer counts.

\mypar{Impact on CCAs}
Bounded queues change how we think about CCAs. For example, they 
allow the design of more aggressive CCAs with \emph{shorter queue targets} that impact less delay-sensitive traffic (mice) (see MSwift \vs UEC's NSCC in \cref{sec:sensitivity}). 

\mypar{Impact on packet sizes} Second, and surprisingly, they change the \emph{optimal packet sizes} in the datacenter networks. 
The impact of packet size as a parameter in AI training is one that, to our knowledge, has not been extensively explored in the networking literature. 
Intuitively, sending $D$ bytes with a packet size $P$ and header $H$, involves a trade-off: large packets increase queue size and hence CCT, while small packets incur higher header overheads. We find that we can model the optimal packet size for some of our LB approaches. 
For \DR schemes such as \name, if we model the  $O(1)$ network queueing as a constant of $\alpha$ packets 
(independent of the number of sent packets), 
then we obtain the following (\cref{sec:pkt}): 

\begin{theorem}
\label{thm:pkt}
To minimize the CCT, the  optimal packet size should carry a payload that scales with $\sqrt{D}$:
\begin{equation} \label{eq:payload}
P-H = \sqrt{\frac{H}{\alpha} \cdot D}
\end{equation}
\end{theorem}
Interestingly, the constant $H/\alpha$ reflects the battle between the two types of CCT overhead: \textit{headers} and \textit{queueing}. We illustrate this result in \cref{sec:sensitivity}.

\section{\pipe Pipe-Based Implementation}
\label{sec:implementation}

\subsection{\pipe}

\mypar{Pipes} Modern switch ASICs are organized into several parallel packet-processing \emph{pipes} (Broadcom) or \emph{slices} (Cisco), each handling a fixed subset of the chip's ports with dedicated stateful hardware. Because ports within the same pipe pass through identical hardware execution stages, the pipeline natively guarantees atomic Read-Modify-Write access to local registers at line rate without cross-pipe locks. Hardware pipelines historically scaled around 4-port blocks per processing engine: Broadcom's Tomahawk~3 and~4 utilize 4 ports of 400G per pipeline (1.6~Tbps per pipe)~\cite{broadcom-bcm56980,wheeler2019tomahawk4}, while 51.2~Tbps generation ASICs such as Cisco's Silicon One G200 and P200 explicitly implement 16 independent slices of $4\times800$G ports (3.2~Tbps per slice)~\cite{cisco-nexus9000-brkdcn3939}. Thus, we focus here on pipes of four ports.

\mypar{\pipe} This pipe-based organization motivates an easier-to-implement variant, \emph{\pipe}: instead of a single RR pointer shared among all bottom ports for a given destination key, as in \name's pipeline above, each pipe keeps its own independent RR pointer per key (\eg with \four, the RR pointer in each pipe is shared by four inputs). This sacrifices some consolidation, and hence some of \name's load-balancing quality, since we are effectively splitting each pointer into several independent, uncorrelated ones, similar in spirit to DRB's many per-flow pointers (\cref{sec:algo}), in exchange for a design that fits directly within existing switch architectures, with no cross-pipe combining hardware required. By default, we assume pipes of 4 ports and use \four, matching current switch ASICs; we explore the impact of this choice on CCT performance in \cref{sec:sensitivity}.

\mypar{Example} For our running example ($\k=64$, so $\w=\k/2=32$ downlink ports) with $8$ pipes of $4$ ports each, \pipe replaces each of \name's $N\approx4K$ shared pointers with $8$ independent, per-pipe pointers, for a total of $8 \times 4K \approx 32K$ pointers in the switch.

\mypar{Bounded queue sizes} 
\pipe keeps the same bounded queueing model as \drb and \name (proof in \cref{sec:pipe}):
\begin{theorem}
\label{thm:pipe}
The modeled average queue size for \pipe is bounded:
\begin{equation} \label{eq:pipe}
    \q = \Theta \para{1}.
\end{equation}
\end{theorem}
\noindent The intuition is that splitting a pointer into several independent ones keeps each of them balanced, so the load per destination on every link stays balanced as well. The price is a larger number of pointers, and hence a larger queue bound constant.

\subsection{P4 implementation}
\begin{table}
\caption{Resource consumption of \pipe on an Intel Tofino 1 switch.}
\label{tab:tofino-resource-consumption}
\centering %
    \begin{tabular}{@{}l c c@{}}
    \toprule
                     & \multicolumn{2}{c}{Usage}                  \\ \cmidrule{2-3}
    Resource         & Absolute        & Relative                 \\ \midrule
    Match Crossbar   & 16 bytes        & 3.7\%                    \\
    Meter ALU        & 1 ALU           & 16.7\%                   \\
    Gateway          & 2 gateways      & 3.6\%                    \\
    SRAM             & 207,530 bytes   & 6.7\%                    \\
    TCAM             & 3,754 bytes     & 5.6\%                    \\
    VLIW Instruction & 3 slots         & 5.6\%                    \\
    Hash Bits        & 82 bits         & 11.0\%                   \\
    \bottomrule
    \end{tabular}
\end{table}

We implement \pipe using Intel P4 Studio~\cite{p4-studio} on the Intel Tofino 1 platform~\cite{tofino} to showcase its practicality. Each pipe stores ($\#\text{edge switches} \cdot \#\text{size classes}$) RR pointers in a single register. The register is indexed by the destination edge switch and size class. Its value at the index is read and updated for each packet.

\Cref{tab:tofino-resource-consumption} summarizes the resource consumption of \pipe on an Intel Tofino 1 switch assuming a network with 64 edge switches and 4 pipes. Most of the SRAM is not used by the RR pointers themselves, but by the table that maps each (destination address, size class) pair to its RR pointer index. This table could be removed by assuming that the IP address space follows the physical topology, \eg with a prefix per edge switch, which our implementation does not assume. The RR pointers themselves grow with the number of pipes and of edge switches, but only in steps, since Tofino allocates SRAM in blocks: the block holding the pointers for 64 edge switches holds them for up to about $1{,}024$ edge switches, and the $2{,}048$ edge switches of our $\k=64$ running example need a single additional block. Since combining several memory operations in a single register is common in Tofino programs, the growth remains at most linear. \pipe thus still leaves plenty of resources for the implementation of other data plane functionality.

Due to resource limitations, we cannot evaluate a hardware cluster of Tofino switches implementing \pipe. We therefore resort to simulation for performance comparison with \sysname in~\cref{sec:sensitivity}.

\section{Evaluation}
\label{sec:sensitivity}

We now determine the robustness of our results when varying several scaling parameters (\S\ref{sec:scaling}), introducing link failures (\S\ref{sec:failures}), changing the topology (\S\ref{sec:topo}), deploying \four incrementally (\S\ref{sec:mix}), using congestion control (\S\ref{sec:cca}) and considering MoE traffic (\cref{sec:moe}). 
Finally, we present results measuring CCT in a scenario simulating the training of Llama-3~\cite{llama3modelcard} models with up to 405B parameters parallelized using FSDP~\cite{rajbhandari2020zero,zhao2023pytorch} and running a full stack of CCA, loss recovery, and load balancing (\S\ref{sec:e2e}).

\mypar{Simulator} We compare different load balancing schemes with simulations on \htsim~\cite{csghtsim}, which we chose due to its ability to scale to large datacenter networks and its existing implementation of many design points in this space. 
Specifically, unless stated otherwise, we simulate a 432-node 3-tier fat-tree network corresponding to $\k=12$-port switches, with $800\unit{Gbps}$ links, $0.5\us$ link latency, RoCE-v2-sized packets with $4,096\unit{B}$ payloads and $62\unit{B}$ headers, $64\unit{B}$ ACKs, and $20\unit{B}$ gaps between frames ($12\unit{B}$ inter-frame gap and $8\unit{B}$ preamble and start frame delimiter), $800\unit{KB}$ per-port
buffers,\footnote{This value roughly follows the per-port per-Gbps buffer size of modern switches~\cite{broadcomtom,spectrumx}, is above a BDP, and resulted in no loss in most evaluated cases. We only use it in the CCA cases, else we use a practically infinite buffer.}, and $80\unit{KB}$ ECN thresholds (corresponding to 10\% of the buffer size; we increase this threshold to 10 packets when studying $64\unit{KB}$ packets). 
We also perform a sensitivity analysis for many of these parameters. 
We run all considered scenarios at least 10 times, and display error bars for the SEM (unless they affect readability).

\mypar{Workloads} 
In our evaluation of LB schemes, we consider an all-to-all (ATA) traffic matrix and a random permutation matrix. 
We focus on these two traffic matrices since they are representative of many collective algorithms, rather than modeling specific implementations of every collective. Permutation matrices are often used to implement collectives such as AllGather and AllReduce \cite{ring-mpich, nccl-ring, sccl, meta} and can be used to iteratively perform an AlltoAll ($n-1$ permutation matrices iteratively). Similarly, an all-to-all traffic matrix may implement AllGathers or a one-shot AllReduce, as well as an AlltoAll.
Notably, in most training jobs, multiple forms of parallelism will be used and thus no one collective will span the whole network. We explore this more in \S\ref{sec:sensitivity} and for now make the simplifying assumption of one collective at a time across all nodes.
We use a default permutation message size of $1\unit{MB}$.

We ran tests with the entire set of our LB schemes from \S\ref{sec:space}. Across our tests, we find that the \emph{relative} ranking and broad assessment of  different schemes that we have made so far remains unchanged. For legibility, in this section, we compare results for  switch-based \spktar and host-based \hpktar, representing  state-of-the-art algorithms, against our \name and \four algorithms. 

\mypar{Metric} Our main metric is the \emph{collective completion time} (CCT), which is the time from when the first packet of a collective is sent until the last packet of that collective is acknowledged. %

\subsection{Scaling}
\label{sec:scaling}

\begin{figure*}
    \centering
    \begin{subfigure}[]{1\linewidth}
    \centering
        \includegraphics[width=0.3\textwidth]{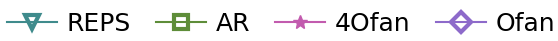}
    \end{subfigure}   
    \par   
    \begin{subfigure}[b]{0.31\linewidth}
        \includegraphics[width=0.9\textwidth]{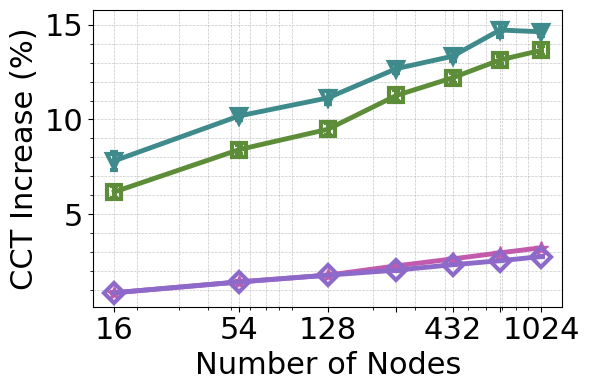}
        \caption{Permutation}
        \label{fig:netw1}
    \end{subfigure}%
    \hfill    
    \begin{subfigure}[b]{0.31\linewidth}
        \includegraphics[width=0.9\textwidth]{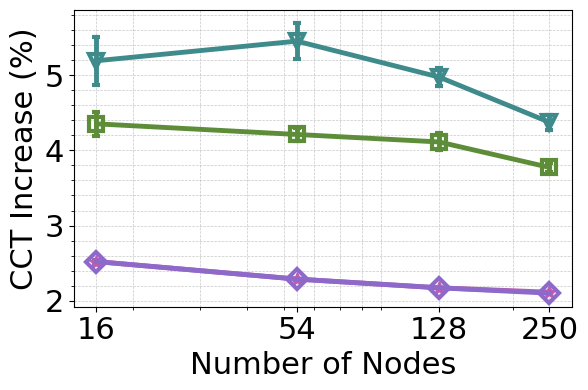}
        \caption{ATA, constant per-flow size}
        \label{fig:netw2}
    \end{subfigure}%
    \hfill     
    \begin{subfigure}[b]{0.31\linewidth}
        \includegraphics[width=0.9\textwidth]{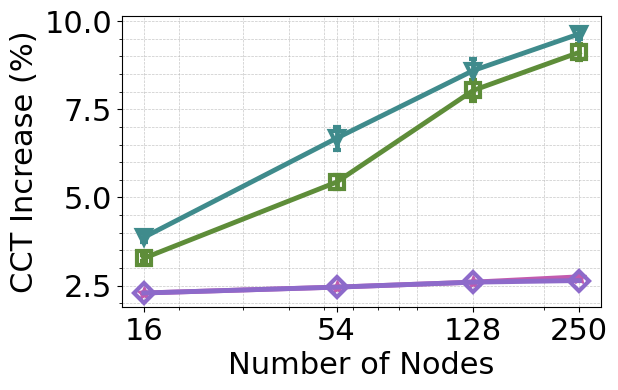}
        \caption{ATA, constant per-host size}
        \label{fig:netw3}
    \end{subfigure}%
    \caption{ CCT increase as  network size increases.}
    \label{fig:networksize}
\end{figure*}

\mypar{Network Size}
\Cref{fig:networksize} shows the impact of scaling the number of nodes in the network (and thus participating in the collective). For ATA, we illustrate two ways of scaling the network: either with constant per-flow size (\cref{fig:netw2}), corresponding to a model that scales with the network; or constant total host sending size (\cref{fig:netw3}), corresponding to adding more nodes with a constant model.
We see that both the relative ordering of the schemes and the magnitude of their CCT performance differences holds or even increases as the network size grows, both for permutation and all-to-all traffic. In particular, in the permutation case, the CCT gap between \four and existing algorithms grows from about 5\% with 16 nodes to about 11\% with 1,024 nodes; and the gap between \four and \name remains reasonably small.

\begin{figure}
    \centering
    \begin{subfigure}[]{1\columnwidth}
    \centering
        \includegraphics[width=0.8\columnwidth]{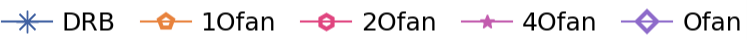}
    \end{subfigure}
    \begin{subfigure}[b]{0.85\columnwidth}
        \includegraphics[width=\columnwidth]{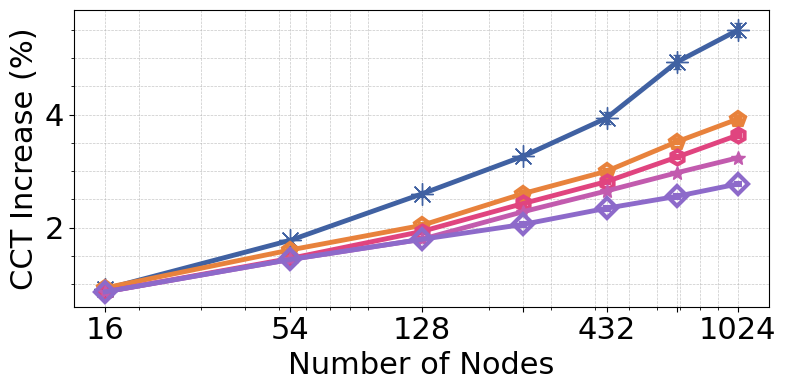}
    \end{subfigure}%
    \caption{CCT increase of \pipe variants vs. DRB as network size increases for a random permutation.}
    \label{fig:vsDRB}
    \vspace{-0.5em}
\end{figure}

\mypar{Pipe Size} \fig{fig:vsDRB} shows the impact of scaling both the pipe size in \pipe and the network size. As network grows, the gap between \pipe and \name predictably expands, but somewhat moderately. In addition, it is interesting to compare \one to \drb: \one holds a per-destination decision at each switch input, while \drb holds a per-destination decision at each sending host. Since each edge switch input corresponds to one host, we would expect them to be similar. However, the plot for \one exhibits and quantifies the power of consolidation at the aggregation switches, which seems to roughly correspond to half the CCT gap between \name and \drb.

\begin{figure}
    \centering
    \begin{subfigure}[]{1\linewidth}
    \centering
        \includegraphics[width=0.6\textwidth]{f/reduced_legend.png}
    \end{subfigure}   
    \par   
    \begin{subfigure}[b]{0.48\linewidth}
        \includegraphics[width=\textwidth]{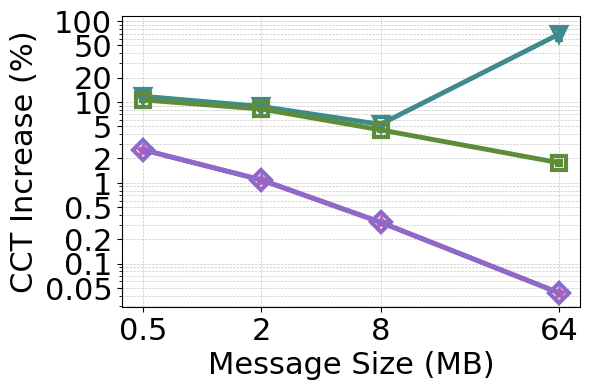}
        \caption{Permutation}
        \label{fig:40d}
    \end{subfigure}%
    \hfill    
    \begin{subfigure}[b]{0.48\linewidth}
        \includegraphics[width=\textwidth]{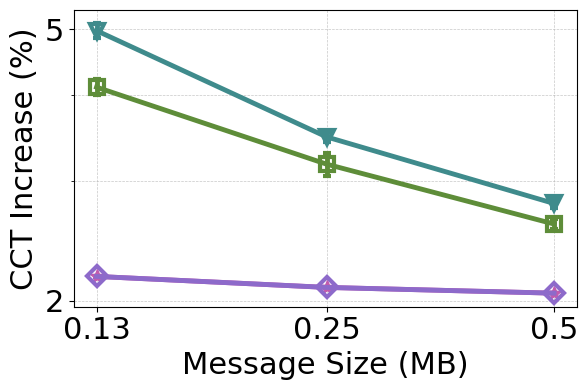}
        \caption{All-to-All}
    \end{subfigure}%
    \caption{CCT increase as  message size increases.}
    \label{fig:messagesize}
\end{figure}

\mypar{Message Size}
\Cref{fig:messagesize} plots the CCT increase as the message size increases. First, the relative ordering of the algorithms does not change, and \four and \name keep beating \reps and \ar. 
Second, the absolute CCT gap between them decreases with increasing message size. This is because, as the message size $\m$ increases, the CCT gap is expected to increase as $\sqrt{\m}$ (\cref{sec:theory}) while CCT grows linearly with $\m$, so the ratio should decrease like $\sqrt{\m}/\m = 1/\sqrt{\m}$.
Third, \cref{fig:40d} illustrates the surprising (and apparently previously-unreported) phenomenon that \reps gets notably worse at very large message size with permutation traffic. We conjecture that this is due to the ``freezing mode'' in \reps triggering as queues become long enough to cause a timeout~\cite{bonatoreps}. %

\begin{figure}
    \centering
    \begin{subfigure}[]{0.9 \columnwidth}
        \includegraphics[width=\columnwidth]{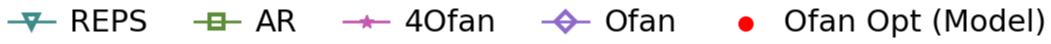}
    \end{subfigure}    
    \begin{subfigure}[t]{0.48\columnwidth}
        \centering
        \includegraphics[width=\linewidth]{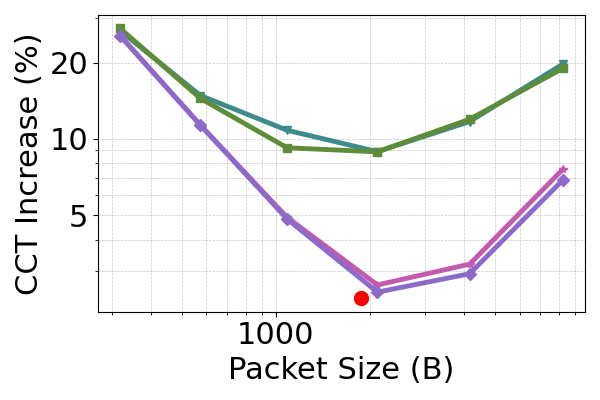}
        \caption{$0.5\unit{MB}$ message}
        \label{fig:pktsize1}
    \end{subfigure}
    \hfill
    \begin{subfigure}[t]{0.48\columnwidth}
        \centering
        \includegraphics[width=\linewidth]{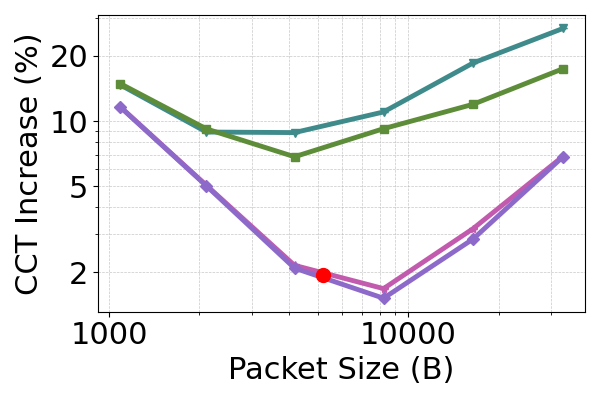}
        \caption{$4\unit{MB}$ message}
        \label{fig:pktsize2}
    \end{subfigure}
    \caption{CCT increase as packet size increases, for two message sizes in bytes.
    }
    \label{fig:pktsize}
\end{figure}

\mypar{Packet Size} 
We illustrate with a red dot in~\cref{fig:pktsize} the optimal size and modeled CCT increase \vs our best constructed lower bound. The graph shows a clear ``sweet spot'' in packet size at which CCT is optimized.  We can see how the model roughly predicts the packet-size sweet spot observed in simulation and its CCT increase. 
\cref{fig:pktsize2} shows how the shorter queue sizes of \name also encourage larger packets: We can derive that when using spraying algorithms with square-root queue-size growth, optimal packet size only grows as $\Theta\para{D^{1/3}}$ rather than $\Theta\para{\sqrt{D}}$. 
Also note that we neglect second-order factors, \eg the CRC may need to be a bit longer to protect longer packets.
This result suggests many avenues for future work on optimizing performance through dynamically tuning packet size, and exploring standardization support for larger Ethernet MTUs.

\subsection{Link Failures}\label{sec:failures}

\begin{figure}
    \centering
    \begin{subfigure}[]{0.7\columnwidth}
    \centering
        \includegraphics[width=0.8\columnwidth]{f/reduced_legend.png}
    \end{subfigure}
    \begin{subfigure}[]{0.7\columnwidth}
        \includegraphics[width=\columnwidth]{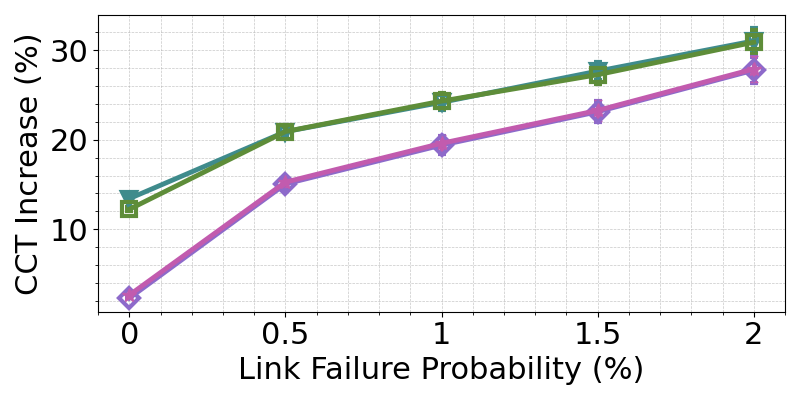}
    \end{subfigure}%
    \caption{CCT increase due to random link failures (beyond the CCT baseline  in a fully-connected network).}
    \label{fig:fail}
\end{figure}

To evaluate the impact of random link failures, we drop each edge-aggregation and aggregation-core link with a given probability in each topology instance. We then look at the increase in the average CCT, compared to the CCT lower-bound baseline without link failures. Note that in each topology instance, given the set of link drops, 
each host paces at a fixed rate, $\rhomax$, computed as the highest rate at which all flows could send if traffic were balanced perfectly across all valid shortest paths, to make sure it's feasible (\cref{sec:approach}).

\Cref{fig:fail} shows the results. We can see that the relative ordering of the LB schemes is kept about the same, and \name keeps outperforming \reps and \ar.

\subsection{Different Topologies}
\label{sec:topo}

\begin{figure}
    \centering
    \begin{subfigure}[]{1\linewidth}
    \centering
        \includegraphics[width=0.7\textwidth]{f/reduced_legend.png}
    \end{subfigure}   
    \par   
    \begin{subfigure}[b]{0.48\linewidth}
        \includegraphics[width=\textwidth]{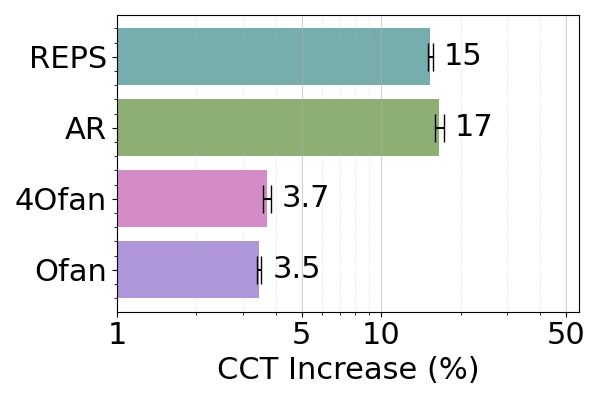}
        \caption{CCT}
    \end{subfigure}%
    \hfill    
    \begin{subfigure}[b]{0.48\linewidth}
        \includegraphics[width=\textwidth]{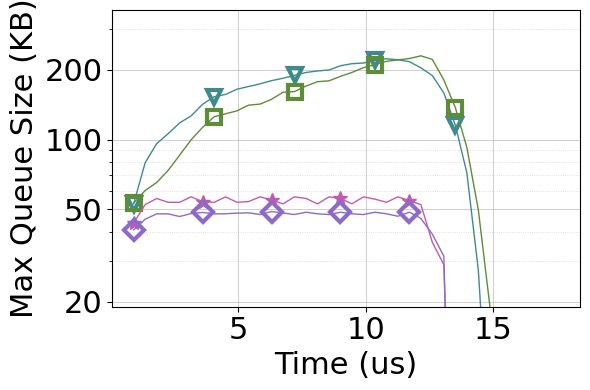}
        \caption{Max queue size}
    \end{subfigure}%
    \caption{Two-level fat-tree with 2,048 hosts.}
    \label{fig:topo}
\end{figure}

\Cref{fig:topo} illustrates the results of running the algorithms in a two-level fat-tree with switch radix $\k=64$ and $\k^2/2=2,048$ hosts. \name and \four keep outperforming \reps and \ar and reduce the CCT by over 11\%.

\subsection{Incremental Deployment}
\label{sec:mix}

\begin{figure}
    \centering
    \begin{subfigure}[]{1\linewidth}
    \centering
        \includegraphics[width=0.7\textwidth]{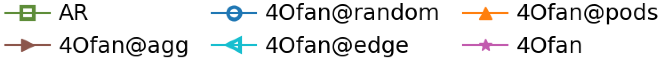}
    \end{subfigure}   
    \par   
    \begin{subfigure}[b]{0.48\linewidth}
        \includegraphics[width=\textwidth]{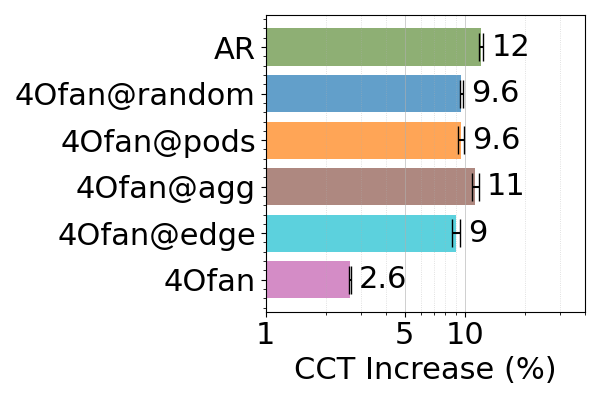}
        \caption{CCT}
    \end{subfigure}%
    \hfill    
    \begin{subfigure}[b]{0.48\linewidth}
        \includegraphics[width=\textwidth]{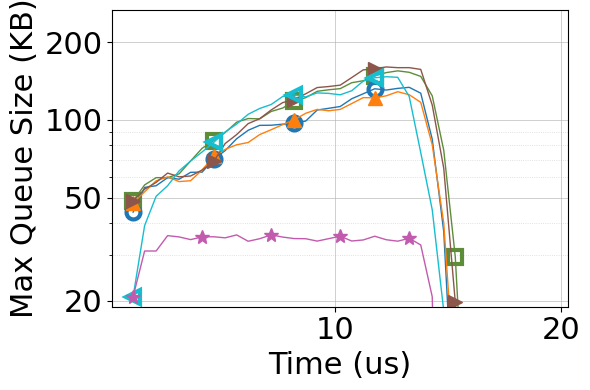}
        \caption{Max queue size}
    \end{subfigure}%
    \caption{Incremental deployment: scenarios where half the switches use \ar and half \four.}
    \label{fig:mix}
\end{figure}

We want to study the impact of incrementally deploying \four in a datacenter using \ar, since operators roll out new switch generations incrementally~\cite{poutievski2022jupiter,zhang2019lifecycle} and may worry about mixing switches with conflicting LB policies. We assume that exactly half of the LB switches (\ie edge and aggregation) use \four, while the other half keeps using \ar. We compare four ways of selecting the switches to deploy \four: (i)~in random switches, (ii)~in random pods (which seems closest to how operators deploy new switches, one aggregation block at a time~\cite{poutievski2022jupiter}), and in (iii)~all aggregation or (iv)~all edge switches.

\Cref{fig:mix} conveys three lessons. 
First, introducing \four in some switches only decreases the CCT, \ie mixing the two LB policies in the network is not harmful. 
Second, the reduction in CCT is less than proportional to the fraction of upgraded switches: the midpoint between the CCT increases of \ar and \four is 7.3\%, and the CCT increase in all mixes is slightly higher. This was expected, since the CCT is set by the slowest flows, which are likely to cross at least one \ar switch, so measuring the worst-case scenario is intrinsically unfavorable to mixes. 
Last, interestingly, the best place to deploy \four is at the edge switches. %

\subsection{Congestion Control }
\label{sec:cca}

\begin{figure*}
\centering
\begin{subfigure}[]{0.32\linewidth}
    \includegraphics[width=\linewidth]{f/reduced_legend.png}
\end{subfigure}  
\par
\begin{subfigure}{0.32\linewidth}
	\includegraphics[width=0.75\textwidth]{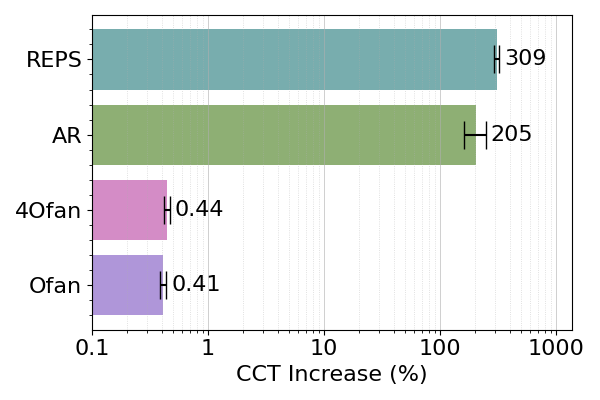}
	\caption{DCTCP: CCT} 
	\label{fig:d1}
\end{subfigure}
\hfill
\begin{subfigure}{0.32\linewidth}
	\includegraphics[width=0.75\textwidth]{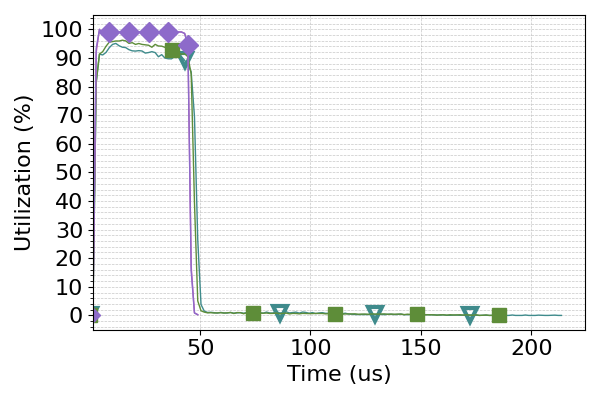}
	\caption{DCTCP: Utilization}
	\label{fig:d2}
\end{subfigure}
\hfill
\begin{subfigure}{0.32\linewidth}
	\includegraphics[width=0.75\textwidth]{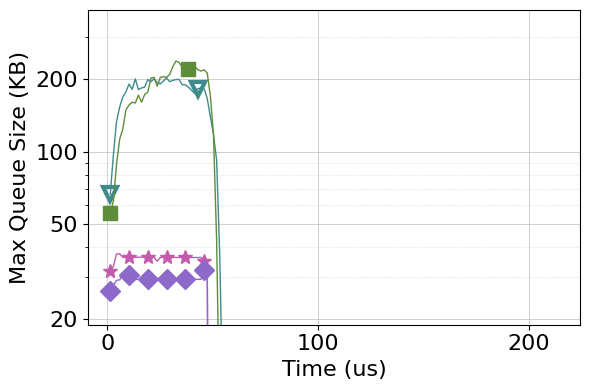}
	\caption{DCTCP: Max Queue Size}
	\label{fig:d3}    
\end{subfigure}

\vspace{2ex} %

\begin{subfigure}{0.32\linewidth}
	\includegraphics[width=0.75\textwidth]{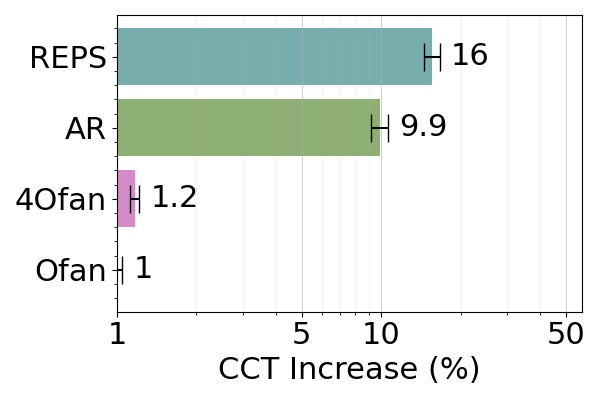}
	\caption{NSCC: CCT} 
	\label{fig:n1}
\end{subfigure}
\hfill
\begin{subfigure}{0.32\linewidth}
	\includegraphics[width=0.75\textwidth]{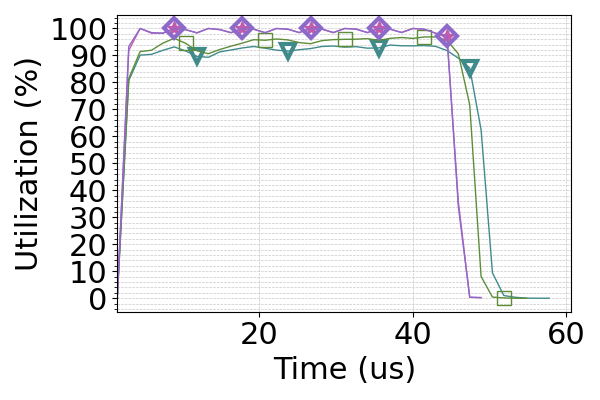}
	\caption{NSCC: Utilization}
	\label{fig:n2}
\end{subfigure}
\hfill
\begin{subfigure}{0.32\linewidth}
	\includegraphics[width=0.75\textwidth]{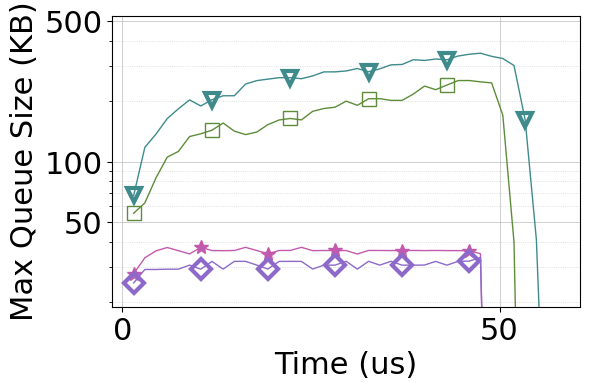}
	\caption{NSCC: Max Queue Size}
	\label{fig:n3}    
\end{subfigure}

\vspace{2ex} %

\begin{subfigure}{0.32\linewidth}
	\includegraphics[width=0.75\textwidth]{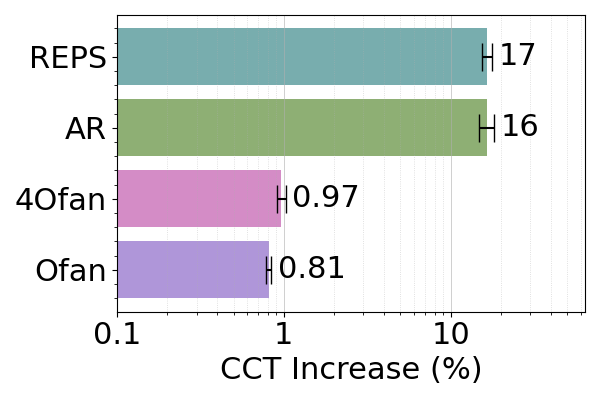}
	\caption{MSwift: CCT} 
	\label{fig:m1}
\end{subfigure}
\hfill
\begin{subfigure}{0.32\linewidth}
	\includegraphics[width=0.75\textwidth]{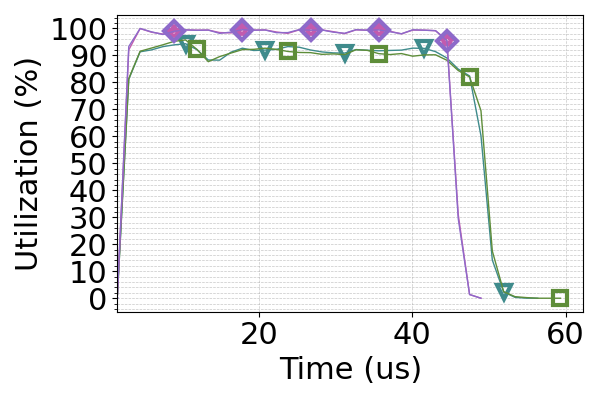}
	\caption{MSwift: Utilization}
	\label{fig:m2}
\end{subfigure}
\hfill
\begin{subfigure}{0.32\linewidth}
	\includegraphics[width=0.75\textwidth]{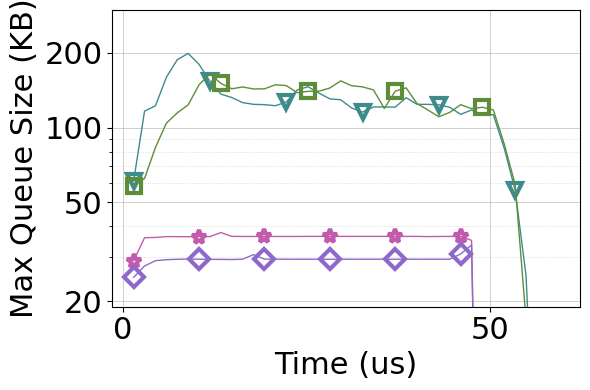}
	\caption{MSwift: Max Queue Size}
	\label{fig:m3}    
\end{subfigure}
\caption{Permutation with  $4$\,MB messages using either (a-c)~DCTCP, (d-f)~UEC's NSCC or (g-i)~MSwift.}
\label{fig:cca}
\end{figure*}

We now compare per-packet spraying algorithms with CCA. We rely on three different CCAs: (i)~DCTCP, a standard datacenter CCA, and two recent CCAs designed for spraying: (ii)~UEC's NSCC~\cite{uec,uec-spec} and (iii)~MSwift~\cite{mswift}. We keep evaluating a permutation traffic, but increase the message size from $1\unit{MB}$ to $4\unit{MB}$ to give enough time for the CCAs to adapt (propagation RTT is about 150 packets, so we increase the message size from 1.7 RTTs to 6.8 RTTs).

\Cref{fig:cca} illustrates the CCA impact. 
With DCTCP, the CCTs of \reps and \ar are just awful, as they are multiplied by $3$--$4\times$, while \pipe and \name have a CCT increase under 1\%. The main reason is that even though htsim's default DCTCP implementation allows for high duplicate-ACK threshold (specifically, 3 plus the number of paths that can be taken by the flow), a few flows still reach the threshold, aggressively decreasing their windows and hurting the CCT. With such a stark difference, we drop DCTCP from now on.

With the spraying-based CCAs, \reps and \ar also perform worse, but by more reasonable margins.
Each CCA needs to rein down the flow rates for \reps and \ar, and therefore their CCT inflation can worsen by at least 8\% with NSCC and 14\% with MSwift \vs \name and \four, which do not need to be slowed down as much. The difference is a bit larger with MSwift as it targets smaller queues than NSCC.

\subsection{MoE Workload}
\label{sec:moe}

\begin{figure}
\centering
\begin{subfigure}{0.49\linewidth}
	\includegraphics[width=\textwidth]{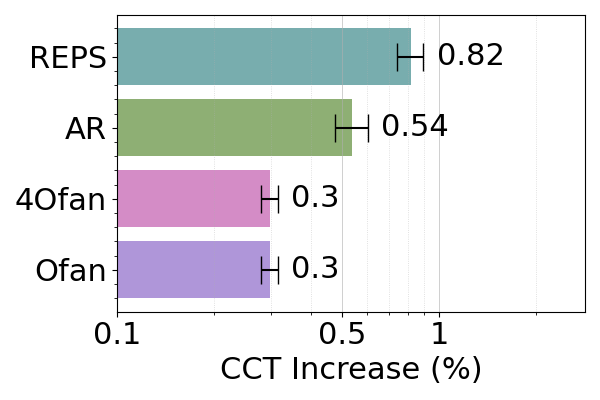}
	\caption{No CCA} 
	\label{fig:moe1}
\end{subfigure}
\hfill
\begin{subfigure}{0.49\linewidth}
	\includegraphics[width=\textwidth]{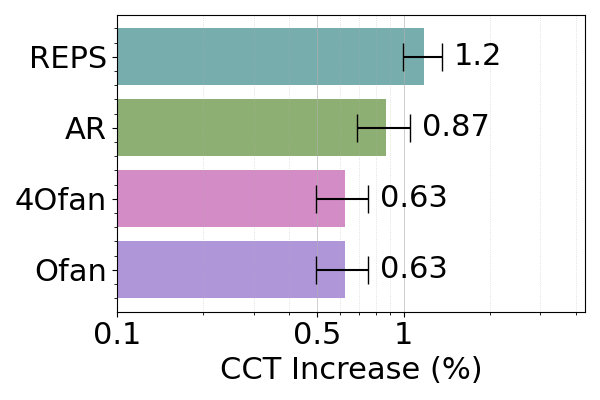}
	\caption{No CCA+EPLB} 
	\label{fig:moe2}
\end{subfigure}

\vspace{2ex} %

\begin{subfigure}{0.49\linewidth}
	\includegraphics[width=\textwidth]{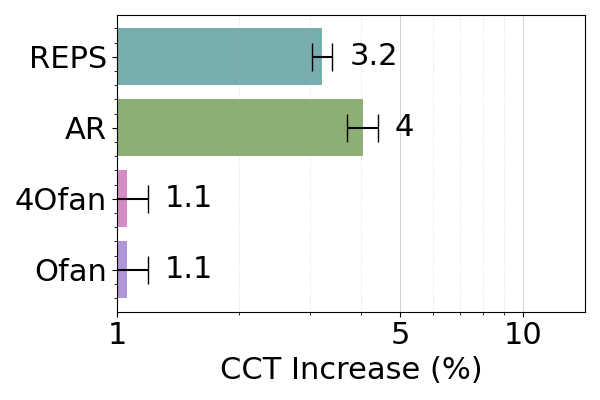}
	\caption{NSCC}
	\label{fig:moe3}
\end{subfigure}
\hfill
\begin{subfigure}{0.49\linewidth}
	\includegraphics[width=\textwidth]{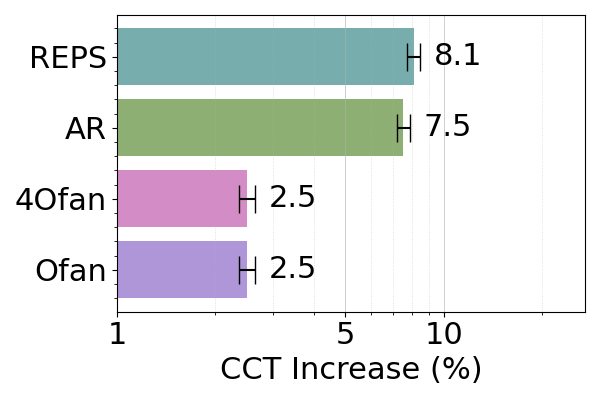}
	\caption{NSCC+EPLB}
	\label{fig:moe4}
\end{subfigure}
\caption{MoE workload: (a-b) without CCA and (c-d) with NSCC, either without or with EPLB reducing the skew.}
\label{fig:moe}
\end{figure}

We generate an MoE routing trace by running one training iteration of the open-source Qwen3-235B model, which has 128 experts per layer and activates 8 experts per token~\cite{qwen3}, starting from its released weights with synthetic input tokens, and recording the token-to-expert routing decisions of one layer.
The trace yields a $128\times128$ GPU-to-expert matrix of token counts. Each token copy carries 8,192\,B for a hidden size of 4,096 in BF16. We then consider a $\k=8$ fat tree with $\n=128$ GPUs. At each run, we randomly map one expert per GPU, and obtain a $128\times128$ GPU-to-GPU traffic matrix. As in DeepSeek-V3's DeepEP~\cite{deepseekai2025deepseekv3technicalreport,deepep}, dispatch is a hierarchical all-to-all over 8-GPU servers: each GPU sends the tokens for a remote server to its same-rank peer there, which forwards them over NVLink, so only inter-server traffic crosses the fat tree. Each GPU sends 492\,MB on average over the fat tree, and the load is skewed: the busiest GPU receives $1.48\times$ the mean (728\,MB). Given these large messages, we use 64\,KB packets, in line with the larger optimal packet size of \cref{thm:pkt}, which also keeps the simulations tractable. 

To evaluate the sensitivity to a less skewed MoE workload, we also replicate popular experts using expert-parallel load balancing (EPLB)~\cite{deepseekai2025deepseekv3technicalreport,eplb}, such that each GPU hosts two expert slots, which reduces the skew to $1.13\times$.

\Cref{fig:moe} illustrates the results. We start with an idealized example without CCA (\cref{fig:moe1,fig:moe2}), where we just allow all sources to blast traffic with infinite buffers and no losses. Since MoE traffic is incast-like and is bound by the most requested GPU, the bottleneck is at the receiver and not in the network. Thus, buffers grow rapidly (beyond $100\unit{MB}$), and the LB impact is small, although \name's CCT is still a bit better. 

We next employ NSCC, and differences become larger, especially with a smaller skew (\cref{fig:moe4}), as LB takes a larger role. \name clearly outperforms other algorithms. Note that since $\w=\k/2=4$ here, \name and \four are equivalent.

\subsection{Training Scenario}
\label{sec:e2e}

\begin{figure*}
    \centering
    \begin{subfigure}[h!]{0.35\linewidth}
        \includegraphics[width=\linewidth]{f/reduced_legend.png}
    \end{subfigure}  
    \par    
    \begin{subfigure}[t]{0.19\linewidth}
        \centering
        \includegraphics[width=\linewidth]{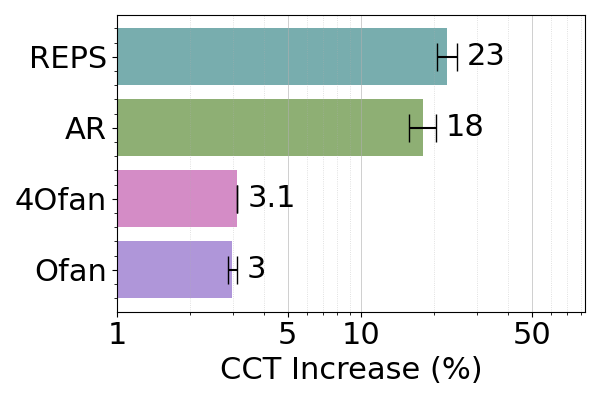}
        \caption{8B, NSCC: CCT}
        \label{fig:F1}
    \end{subfigure}
    \hfill
    \begin{subfigure}[t]{0.19\linewidth}
        \centering
        \includegraphics[width=\linewidth]{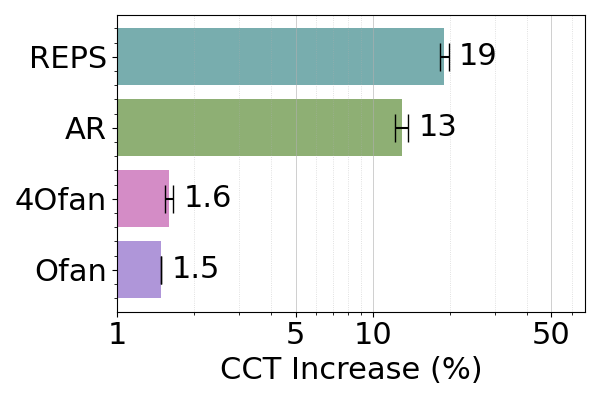}
        \caption{70B, NSCC: CCT}
        \label{fig:F2}
    \end{subfigure}
    \hfill
    \begin{subfigure}[t]{0.19\linewidth}
        \centering
        \includegraphics[width=\linewidth]{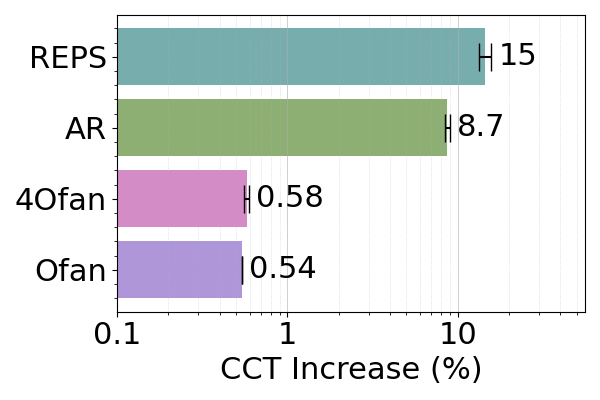}
        \caption{405B, NSCC: CCT}
        \label{fig:F3}
    \end{subfigure}
    \hfill    
    \begin{subfigure}[t]{0.19\linewidth}
        \centering
        \includegraphics[width=\linewidth]{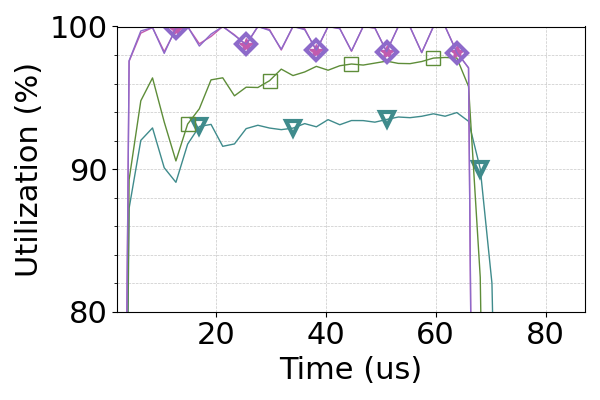}
        \caption{405B, NSCC: Util.}
        \label{fig:F4}
    \end{subfigure}    
    \hfill
    \begin{subfigure}[t]{0.19\linewidth}
        \centering
        \includegraphics[width=\linewidth]{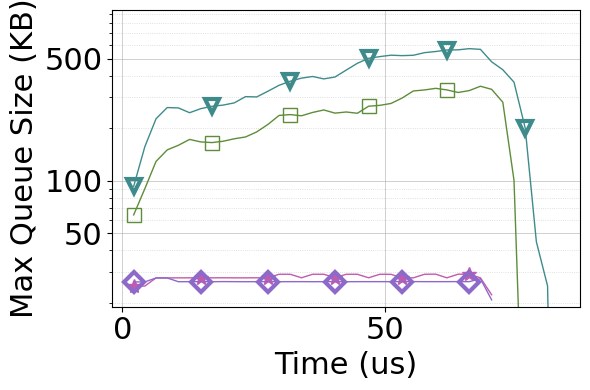}
        \caption{405B, NSCC: Queue}
        \label{fig:F5}
    \end{subfigure}

\vspace{2ex} %

    \begin{subfigure}[t]{0.19\linewidth}
        \centering
        \includegraphics[width=\linewidth]{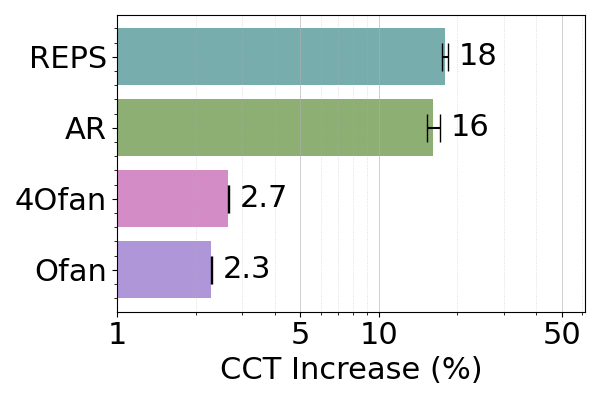}
        \caption{8B, MSwift: CCT}
        \label{fig:F6}
    \end{subfigure}
    \hfill
    \begin{subfigure}[t]{0.19\linewidth}
        \centering
        \includegraphics[width=\linewidth]{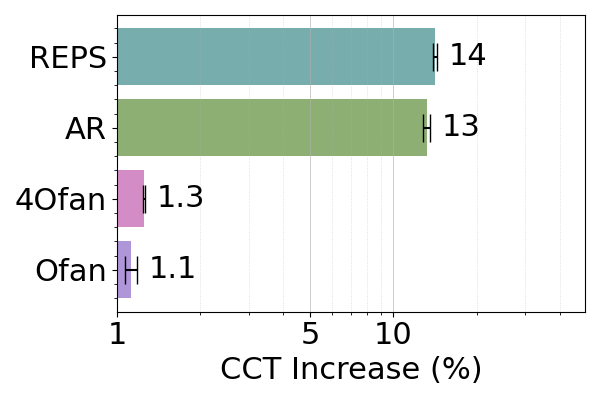}
        \caption{70B, MSwift: CCT}
        \label{fig:F7}
    \end{subfigure}
    \hfill
    \begin{subfigure}[t]{0.19\linewidth}
        \centering
        \includegraphics[width=\linewidth]{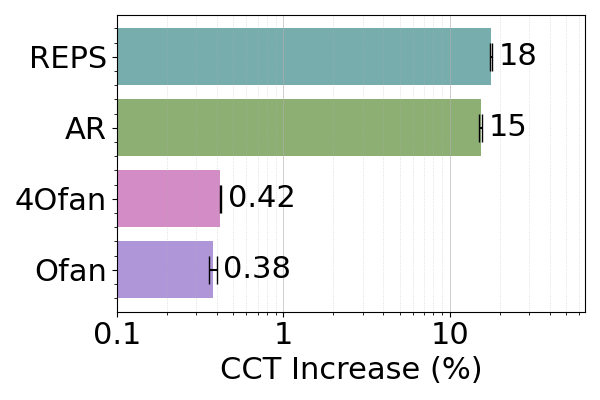}
        \caption{405B, MSwift: CCT}
        \label{fig:F8}
    \end{subfigure}
    \hfill    
    \begin{subfigure}[t]{0.19\linewidth}
        \centering
        \includegraphics[width=\linewidth]{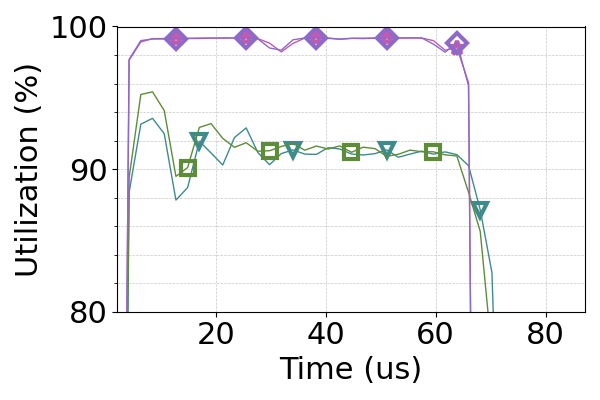}
        \caption{405B, MSwift: Util.}
        \label{fig:F9}
    \end{subfigure}    
    \hfill
    \begin{subfigure}[t]{0.19\linewidth}
        \centering
        \includegraphics[width=\linewidth]{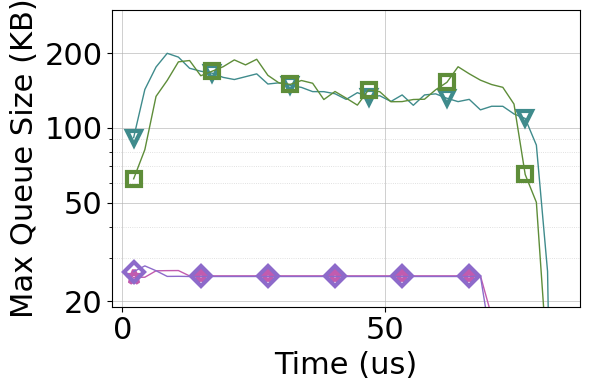}
        \caption{405B, MSwift: Queue}
        \label{fig:F10}
    \end{subfigure}
    \caption{FSDP2 with Llama-3 8B, 70B and 405B parameters using either (a-e)~NSCC or (f-j)~MSwift.
        }
    \label{fig:FSDP}
\end{figure*}

Finally, we evaluate the performance of each LB scheme on a realistic AI training scenario. 
Specifically, we simulate training Llama-3~\cite{llama3modelcard} (8B, 70B, and 405B) parallelized using Fully Sharded Data Parallel (FSDP2)~\cite{rajbhandari2020zero,zhao2023pytorch,fsdp2} with HSDP configuration on a 1,024-GPU cluster.
The cluster is organized as 128 servers of 8 GPUs, and we assume random server placement. 
We evaluate the main network bottleneck in FSDP2 training by focusing on the backward pass, as it combines a ReduceScatter of gradients and an AllGather of weights.
FSDP2 traffic is implemented as a hierarchical ring, where intra-server communications are performed in the high-bandwidth domain (scale-up, like NVLink), while our fat-tree network implements inter-server communications (scale-out) using 8 parallel rings: each logical GPU $i$ is connected to logical GPU $i+8 \pmod{1,024}$. Because sharding is $1{,}024$-way while each ring spans only $128$ GPUs, a ring step carries one $1/1{,}024$ shard rather than $1/128$ of a layer; the intra-server AllGather then completes the remaining $7/8$. 
We assume FP8 precision and $4\unit{KB}$ packet payload, translating to 119 packets per flow per layer in Llama-3 8B (32 layers), 418 in Llama-3 70B (80 layers) and 1533 in Llama-3 405B (126 layers). 
We enable both the NSCC and MSwift CCAs with loss recovery, and  
measure the CCT of each collective phase.

As shown in \Cref{fig:FSDP}, while \reps and \ar overheads remain about constant, the relative CCT increase for \name drops as model size scales. When compared against the best out of the two existing models, at 405B parameters, \name (resp. \four) reduces the CCT increase beyond baseline by $16\times$ (resp. $15\times$) with NSCC and $39\times$ (resp. $36\times$) with MSwift. This improvement stems from superior utilization (\cref{fig:F4,fig:F9}): Since the CCA must throttle \reps and \ar to maintain stable queue sizes (\cref{fig:F5,fig:F10}), \name achieves significantly higher effective throughput.

\section{Discussion}
\label{sec:discussion}

\mypar{Realizing \name's pointers at line rate} \S\ref{sec:implementation} presented \pipe, the variant that sidesteps cross-pipe coordination by keeping an independent pointer per pipe. It is fair to ask what it would take to implement \name's shared pointers exactly, without that approximation. The difficulty is concurrent access: a switch processes many packets per cycle, and several of them may be bound for the same destination switch and therefore contend for the same round-robin register. This is a classical problem with many processors issuing fetch-and-add to a shared counter. Its  classical solution, a \emph{combining network}~\cite{kruskal1988efficient,lebeck1994request}, merges duplicate requests before they reach the counter. The general machinery is too costly and too variable in latency for a switch pipeline, but our problem is far narrower: a small, fixed set of local pointers. In \cref{sec:implementation-details} we show that a fixed-latency, four-stage pipeline suffices. It sorts the batch by pointer key, combines packets sharing a key, updates each key's register once by the group's count, then decombines and unsorts so each packet leaves with its own port. A back-of-envelope estimate puts this at roughly $10^5$ gates and a few milliwatts, against a switch ASIC that draws hundreds of watts: negligible. 
We therefore regard \name's exact form as implementable, with \pipe as the pragmatic near-term option.

\section{Related Work}
\label{sec:related}

There is a vast literature on load balancing in datacenters, and we cannot hope to summarize it in any depth. The citations in \cref{tab:models} provide a large subset. 
Many proposed systems \cite{letitflowlet, kandulaflowlet, plb, meta} do not seriously consider packet spraying given its potential adverse effects on the transport due to reordering. Some systems do include spraying in their design \cite{ndp, eqds, dcpim,pfabric}, but often introduce specialized mechanisms to make spraying work (\eg proposing an entirely new transport). Two exceptions are \cite{rps} and \cite{bonatoreps}, where packet spraying is evaluated against ECMP and some other load balancing granularities, but one transport is assumed and the distance from optimality is unclear. 

We already discussed throughout the paper many LB schemes that we compare against: switch- and host-based spraying in \cref{sec:intro} and \cref{tab:models}, and the destination-based \drb and \ethereal in \cref{sec:dest-aware}.

\mypar{Fat-tree routing in HPC} Routing in fat trees has long been studied for HPC interconnects~\cite{duato650interconnection}. Deterministic routing sends all the traffic to a given destination along a fixed path, as in ECMP, without packet spraying. In the paper we focus more on packet spraying, which can achieve better load balancing.

\mypar{Load-balanced routers} Destination-based rotation has roots in three-stage load-balanced router algorithms such as UFS, FOFF and Padded Frames~\cite{LBR,thesis,padded-frames}. Load-balanced routers spread the traffic of each flow cyclically over intermediate ports to guarantee full throughput.

\mypar{\sglb} Broadcom offers a Global Load Balancing (GLB) feature in its Tomahawk chips~\cite{bytedance-glb,juniper-glb}. However, GLB is limited to 3-stage Clos networks with at most 64 edge switches, and cannot run end-to-end on a 5-layer fat-tree topology~\cite{juniper-glb}. This is because GLB relies on the "Next-Next-Hop" (NNH) visibility. This inherently limits the decision scope. For example, an edge switch in a 5-layer fat-tree cannot see the congestion status of the link between a core switch and a distant aggregation switch. ByteDance's \sglb introduces several ways to mitigate this limitation in three-level fat trees~\cite{bytedance-glb}, which we apply in our evaluations.

\ifblind
\else
    \section*{Acknowledgments}
    We thank Nandita Dukkipati, Ziming Mao and David Tennenhouse for their helpful comments. This work was partly supported by the Louis and Miriam Benjamin Chair in Computer-Communication Networks.
\fi

\section{Conclusion}

AI training pushes every source to $100\%$ of line rate, leaving no slack to hide imperfect load balancing: queues grow with the message size and utilization stalls. Because a fat tree's downward path to a destination is unique, destination-oblivious schemes cannot undo the imbalance they create on the way up, as we proved for randomized spraying using an $\Omega(\sqrt{\m})$ queueing lower bound. %

\name instead rotates per destination in the switch, consolidating all traffic sharing a destination onto one pointer. Moreover, we showed how its \four variant fits commodity ASICs 
and cuts CCT inflation by $15$--$36\times$ in an end-to-end FSDP2 evaluation at 405B parameters.

\bibliographystyle{plain-etal} %
\bibliography{references}

\appendix{}

\section{Permutation CCT Lower-Bound}\label{sec:lowerbound}

\mypar{Outline} Our goal is to provide a tight lower-bound baseline for the CCT of a permutation workload when hosts send at 100\% of the link rate. This baseline is used in all of the figures that study the CCT increase caused by a load-balancing scheme beyond this baseline.

\mypar{Example} With our default assumptions (\cref{sec:sensitivity}), we get below a lower bound of $17.05694\us$ for a message size $\m=256.$ We want to compare that to smallest value that we can get. We only send a flow from the first host to the last and back in a small network ($\k=4$), using all algorithms, iterate 1000 times, and take the minimum over all possibilities, we measure in practice a value of $17.0587\us$ (with $\k=8$, we get $17.0609\us$, which is less tight). That is, the difference is $10^{-4}$ of the CCT, less than a packet time, showing the tightness of the lower bound.

\mypar{Dynamics} \Cref{fig:ACK} illustrates how permutation collectives in fat trees exhibit interesting interactions between data and ACK packets, as each host is at the same time a sender and a receiver. When computing the CCT, we find that there are in fact three successive modes of communication: (1)~Data packets only, then (2)~interleaved data and ACK packets, and finally (3)~ACK packets only. In order to provide a lower bound for the CCT, we need to take these three modes into account. 

\begin{figure}
    \centering
    \includegraphics[width=0.65\linewidth]{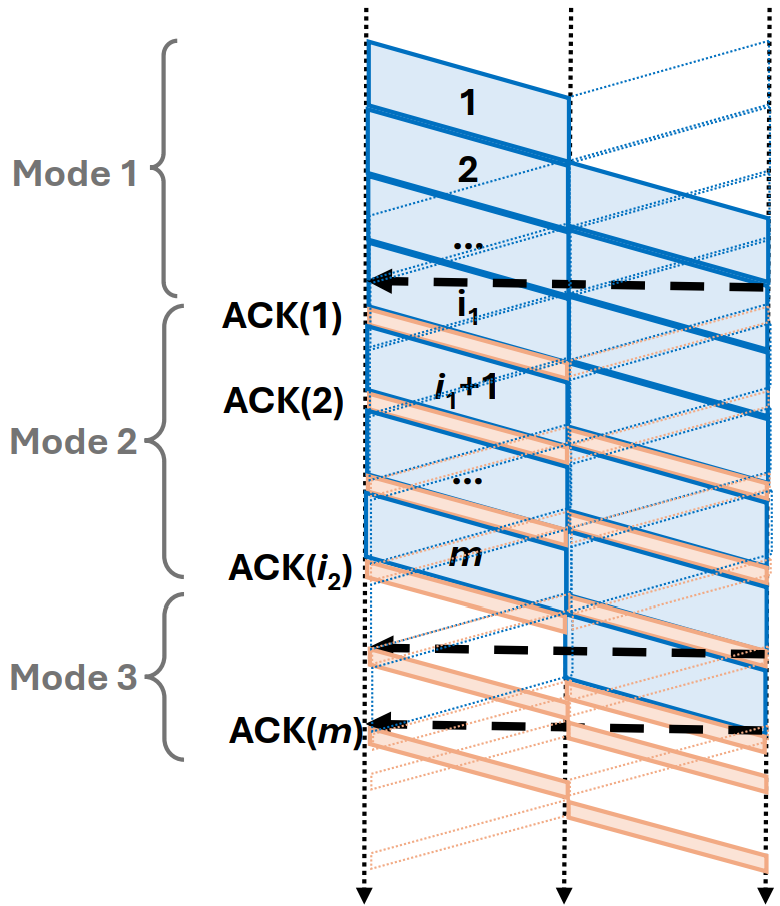}
    \caption{Illustration of the symmetric dynamics of the permutation collective. 
    }
    \label{fig:ACK}
    \vspace{-1em}
\end{figure}

\mypar{Assumptions} In the remainder, we focus on a node that receives a flow of data packets going through $\H=6$ hops, and sends another flow of ACK packets traversing $\H$ hops as well. Let $\Td$, $\Ta$, $\Tg$ and $\p$ respectively denote the data and ACK packet serialization times, the gap time between packets, and the one-way propagation time. Also, let $\d=\Td+\Tg$ and $\a=\Ta+\Tg$.

\mypar{Mode 1: only data} The last bit of the first data packet of each message goes through at least a propagation time of $\p$ and $\H$ serialization delays of $\Td$, thus packet 1 arrives at time $t_1=\p+\H\cdot\Td$.

When it gets to the receiver, this receiver is in the middle of sending data packet indexed $i_1$. It needs $\Td+(i_1-1)\cdot \d$ to send these $i_1$ data packets. Therefore, putting both equations together yields $i_1$:

\begin{equation}
    i_1 = \ceil{\frac{\p+(\H-1)\cdot\Td}{\d}}+1
\end{equation}

\mypar{Example} With the paper's settings, $i_1=78$ data packets. In other words, each sender sends at least 78 consecutive data packets before it  sends ACKs. Of course, if $\m<i_1$, then it only sends $\m$ data packets.

\mypar{Mode 2: interleaved data and ACK packets} Once the host starts interleaving data and ACK packets, we assume that the host policy performs round-robin between the two, as long as there is data left. 

Let's study when the first ACK arrives. Before, we saw that $i_1$ data packets are sent. We also saw that packet 1 arrives at time $t_1=\p+\H\cdot\Td$. Therefore, packet $i_1$ arrives at time $t_{i_1}=t_1+(i_1-1)\d$. The first ACK for packet 1 arrives just after, at $t^{\text{ACK}}_1=t_{i_1}+\a$.

Afterwards, all packets are sent in an interleaved way, data-ACK-data-ACK-etc. Therefore, the ACK for packet $i$ arrives at $t^{\text{ACK}}_i=t^{\text{ACK}}_1+(i-1)(\d+\a)$. 

\mypar{Mode 3: only ACKs} At the end, the hosts send the $\m\th$ packet, and don't need to interleave data and ACKs anymore. This is a symmetric case from $i_1$: it happens as they send ACK $i_2$, with $i_2=\m-i_1+1$. For example, with $\m=256$ and $i_1=78$, it happens at the ACK for packet $179$.

Beyond index $i_2$, we have two constraints. First, ACKs should have time to be transmitted. We get a first lower-bound time $t^{\text{ACK}}_i(1)=t^{\text{ACK}}_{i_2}+(i-i_2)(\a)$. Second, ACK packets can obviously only be sent if they acknowledge a data packet that has arrived. Just like packet $i_1$ arrived just before ACK $1$, so does packet $i$ arrive just before ACK $i-(i_1-1)$. Thus, we obtain a second lower-bound time: $t^{\text{ACK}}_i(2)=t^{\text{ACK}}_{i-(i_1-1)}+(\H-1)(\a)+\p$. 
Finally, the lower bound is the maximum of both constraints: for $i_2+1 \leq i \leq \m$,$
t^{\text{ACK}}_i=
\max\para{t^{\text{ACK}}_i(1),
     t^{\text{ACK}}_i(2)}
$.

\section{Proof of \cref{thm:sqrt}: Square-Root Scaling} \label{sec:sqrt}

\mypar{Proof overview} We consider an edge uplink and an aggregation-to-edge downlink. In both cases, the number of packets that reach the queue is a binomial variable, \ie a sum of independent coin flips, whose mean equals the service. Its excess over the service, and therefore the queue, grows as $\sqrt\m$. We also prove that no queue grows faster than $\sqrt\m$, since every link has load at most 1.

\begin{proof}[Proof of \cref{thm:sqrt}]

\T{No internal traffic.} Since the permutation is a uniform random derangement, the flow destination from any sending host $h$ is uniform over the $\n-1$ other hosts, and it lies under $h$'s own edge switch with probability $\frac{\w-1}{\n-1}$. A union bound over the $\w$ hosts of an edge switch shows that it has no internal traffic with probability
\begin{equation}
    P_{\text{northbound}} \;\geq\; 1-\w\cdot\frac{\w-1}{\n-1} \;=\; 1-\frac{\w(\w-1)}{2\w^3-1} \;\geq\; \frac{13}{15}
    \label{eq:northbound}
\end{equation}
for any fat tree with $\w\geq2$, 
since the middle term is largest at $\w = 2$.

\T{Binomial counts.} Fix an index $\ell\in\{1,\ldots,\w\}$. First, consider uplink $\ell$ of an edge switch $E$, and let $A_t$ be the number of packets it receives among the first $t$ packets of each host. If $E$ has no internal traffic, its $\w$ hosts send all their packets northbound, and each packet independently takes uplink $\ell$ with probability $1/\w$. 

Second, consider the downlink from the $\ell$-th aggregation switch of a pod to an edge switch $F$ of this pod. A packet that leaves its source edge switch through uplink $\ell$ reaches the $\ell$-th aggregation switch of the destination pod, so the downlink receives exactly the packets that the $\w$ sources sending to hosts under $F$ send through uplink $\ell$. Let $A_t$ count them among the first $t$ packets of each source. If $F$ has no internal traffic, each of these packets again takes uplink $\ell$ independently with probability $1/\w$. 

In both cases, $A_t$ is binomial with $\w t$ trials and parameter $1/\w$, with mean $t$ and variance $V=t(1-1/\w)\ge t/2$.
For a binomial variable with $N$ trials and parameter $p\in[1/N,1-1/N]$, the mean absolute deviation is at least $1/\sqrt2$ times the standard deviation~\cite{berend2013sharp}. Thus $E\para{|A_t-t|}\ge\sqrt{V/2}\ge\sqrt t/2$, and since $E[A_t]=t$, $E[(A_t-t)^+]=\frac{1}{2}E\para{|A_t-t|}\ge\sqrt t/4$.

\T{Queues.} The queue serves at most one packet per packet duration, so after $t$ packet durations it holds at least $(A_t-t)^+$ packets, provided that the $A_t$ packets have arrived. This is exact for the edge uplink, which receives its packets directly from the hosts; for the downlink, it neglects upstream queueing delays. By \cref{eq:northbound}, the expected queue after $t$ packet durations is thus at least $\frac{13}{15}\cdot\frac{\sqrt t}4\ge\frac{\sqrt t}5$, and its time average over the $\m$ packet durations is at least $\frac1{5\m}\sum_{t=1}^\m\sqrt t\ge\frac1{5\m}\int_0^\m\sqrt x\,dx=\frac2{15}\sqrt\m$. Edge uplinks and aggregation-to-edge downlinks each form one fifth of the $5\n$ switch queues, leading to the $\sqrt\m$ lower bound.

\T{Upper bound.} Fix a queue $X$, and let $S_b$ count the packets routed through $X$ among the first $b$ packets of each source. Given the permutation, each packet is routed through $X$ independently, and in each packet duration the expected number of such packets is at most one, \eg $\w$ sources with probability $1/\w$ for an aggregation-to-edge downlink. Thus $S_b-S_a\le(b-a)+M_b-M_a$, where $M_b$ is a sum of $b$ independent centered terms of variance at most $1$, so $E[\max_{b\le\m}|M_b|]\le2\sqrt\m$ by Doob's maximal inequality~\cite{durrett2019probability}. Neglecting upstream queueing delays, packets reach $X$ %
after a fixed time, so by Lindley's formula the queue of $X$ is at most $\max_{a<b}(S_b-S_a-(b-a))+1\le2\max_{b\le\m}|M_b|+1$. Hence $\qup=\bigO{\sqrt\m}$ and $\qdown=\bigO{\sqrt\m}$.
\end{proof}

\section{Proof of \cref{thm:sync}: Synchronization} \label{sec:sync}

\T{Intuition for synchronization.}
We prove the synchronization properties of \simplerr and \jsq. The synchronization derives from three properties:

\mypars{1. Perfect match for first packets} Consider an edge switch without internal traffic.
Assume that flow $1$ sends the first packet, which is forwarded to uplink port $i$. Each of the
other flows then sends one packet that is forwarded to a distinct uplink port: (i)~In RR, because of the round-robin rotation; and (ii)~in AR, because the queues of these ports are initially empty, then any port that is assigned a packet gets a queue size of 1 and is removed from consideration for a duration $T$. At the end, there is a perfect 1-to-1 match of packets to uplink ports.

\mypars{2. Sticky flows} When the second packet of flow $1$ arrives after a period $T$: (i)~in $RR$, it will be sent again to port $i$ due to the round-robin periodic rotation, and (ii)~in $AR$, the first packet of flow 1 will have just left, since it takes $\pkt/\C<T$ to service it. Thus port $i$ is the only empty port and the second picks it. In both cases, port $i$ is a sticky port for flow 1. Likewise, each other flow will have its own sticky port.

\mypars{3. Aggregation switch layout}
Now we understand how flows can be stuck onto specific edge-aggregation links. In addition, in a fat tree, a flow stuck to aggregation switch $i$ of a source pod will also be stuck to aggregation switch $i$ of a destination pod (black circles in \cref{fig:cong}). Thus, two flows from different source pods but destined to the same edge switch may have picked the same aggregation switch $i$, leading to contention at the link from the destination aggregation switch to the destination edge switch and to linear growth of its queue.

\mypar{Proof outline} We first bound the probability that a flow is sticky (Step~1). We then show that an aggregation-to-edge downlink shared by several sticky flows sees its queue grow linearly (Step~2). Since sticky flows from distinct source edge switches pick their aggregation index uniformly and independently, a constant fraction of these downlinks receive at least two sticky flows (Step~3), which yields $\q=\Theta(\m)$ (Step~4). Finally, we show that the uplink queues stay bounded (Step~5).

For simplicity, the proof assumes that each flow starts with a random jitter uniformly distributed over a packet sending time.

\begin{proof}[Proof of \cref{thm:sync}]
Throughout, sources are paced at a fixed rate and are not slowed by congestion control (\S\ref{sec:approach}); and buffers are assumed infinite.
Recall that $\w=\k/2$ is the number of aggregation-to-edge downlinks entering an edge switch.

\T{Step 1: stickiness probability.} We lower-bound the probability $p$ that a flow is \emph{sticky}, \ie sends all $\m$ of its packets through a single aggregation index. We count a flow as sticky only when it satisfies the properties below (the second one only for \jsq), which can only underestimate the queue.

\mypars{Property 1: {no internal traffic}} {The flow's source edge switch has no internal traffic, \ie all its $\w$ flows are northbound. The number of northbound flows then equals the number of northbound ports, so the first round of packets forms a perfect matching of flows to ports. This happens with probability $P_{\text{northbound}}\geq 13/15$ by \cref{eq:northbound}.} 

\mypars{Property 2: safe flows}
For \jsq, stickiness also requires that each packet of the flow finds exactly one port with the strictly shortest queue, namely the port used by its previous packet. Let $T_{SERVICE}=T_{DATA}+T_{IPG}$ be the period, where $T_{DATA}$ is a packet's transmission time and $T_{IPG}$ the inter-packet gap. We call the flow \emph{safe} if no other flow of its edge switch completes a packet transmission within $T_{IPG}$ before or after the flow's own packet completes. Otherwise, the competitor's queue empties either just before the competitor's next packet arrives, which may then pick our flow's port among two empty ones, or just before our flow's next packet arrives, which may then pick the competitor's port. When the flow is safe, its next packet finds a single empty queue: its own previous port, and one packet in each other northbound port.
Thus, \jsq then picks the same port deterministically. Since all flows are periodic, the same holds for every subsequent packet. Each of the $\w-1$ other flows, whose jitter is independent and uniform over the period, completes a packet within this window of width $2T_{IPG}$ with probability $2T_{IPG}/T_{SERVICE}$, so
\begin{equation}
    P_{\text{safe}} = \left(1-\frac{2T_{IPG}}{T_{SERVICE}}\right)^{\w-1}.
\end{equation}

\mypars{Stickiness probability} The two properties are independent, since the first depends only on the permutation and the second only on the hosts' jitters. Hence $p=P_{\text{northbound}}\cdot P_{\text{safe}}$ for \jsq. For \simplerr, the round-robin rotation does not depend on queue occupancies, so safety is automatic, $P_{\text{safe}}=1$, and $p=P_{\text{northbound}}$. In both cases, the index at which a flow sticks is uniform over the $\w$ aggregation switches, since it is set by the random initial pointer or tie-breaking together with the hosts' jitters.

\T{Step 2: queue on one downlink.} If $c$ sticky flows share an aggregation-to-edge downlink $\ell$, each emits one packet per period $T$, so by the end of period $j$ the link has received at least $c(j-1)$ packets while serving at most $j$ packets, yielding $Q_\ell(j)\ge(c-1)j-c$. Averaging over $j=1,\ldots,\m$, the time-average of $\ell$'s queue is at least $(c-1)\m/2-\bigO{1}$; the drain phase after the sources stop only increases it. This matters only if a \emph{constant fraction} of the queues carry $c\ge2$, since one overloaded link among $\Theta(\n)$ queues contributes only $\Theta(\m/\n)$ to the average.

\T{Step 3: a constant fraction of downlinks is overloaded.} Let $Y_\ell$ be the number of sticky flows on downlink $\ell$, and fix a destination edge switch $E_D$. Under a permutation exactly $\w$ flows are destined to hosts under $E_D$, one per host. Their $\w$ sources share an edge switch with probability at most $\binom \w2\frac{\w-1}{\n-\w}
<\frac14$ by a union bound over pairs, so they sit at $\w$ distinct source edge switches with probability at least $3/4$ for every $\w\ge 2$ ($\k\ge4$). 

Let's condition on that event. Distinct source switches then decide independently, so the number $B$ of these flows that are sticky is $\mathrm{Bin}(\w,p)$, and each sticky flow falls on one of the $\w$ downlinks into $E_D$ uniformly and independently, its index being fixed by its own switch's jitter and initial state. 
Let $O$ denote the number of occupied downlinks into $E_D$, so that $\sum_{\ell\to E_D}(Y_\ell-1)^+=B-O$. Each of the $\w$ flows independently ends up as a sticky flow on a given downlink $\ell$ with probability $p/\w$, so $Y_\ell\sim\mathrm{Bin}(\w,p/\w)$ and $E[O]=\sum_{\ell\to E_D}P(Y_\ell\ge1)=\w\para{1-\para{1-\tfrac p\w}^{\!\w}}$. Since $E[B]=p\w$,
    \be
    \label{eq:excessED}
    \begin{aligned}
    E\Big[\sum_{\ell\to E_D}(Y_\ell-1)^+\Big]
    &= E[B]-E[O] = p\w-\w\para{1-\para{1-\tfrac p\w}^{\!\w}}\\
    &= \w\brac{\para{1-\frac p\w}^{\!\w}-(1-p)} .
    \end{aligned}
    \ee
    By Bernoulli's inequality $(1-p/\w)^{\w}>1-p$ strictly whenever $\w\ge2$ and $p\in(0,1)$, so the bracket is a positive constant depending only on the topology; it tends to $p-1+e^{-p}$ as $\w$ grows. {\Cref{eq:excessED} holds conditionally on the $\w$ sources sitting at distinct edge switches; since that event has probability at least $3/4$ and the sum is nonnegative, the unconditional expectation is at least $\frac34$ times it.} Summing over the $\k^2/2$ destination edge switches and using $\tfrac{\k^2}{2}\cdot \w=\tfrac{\k^3}{4}=\n$,
    \be
    E\Big[\sum_{\ell} (Y_\ell-1)^+\Big]\;\ge\;\frac34\,\n\,\brac{\para{1-\frac p\w}^{\!\w}-(1-p)}\;=\;\Omega(\n).
    \ee

\T{Step 4: combine.} A three-level fat tree has $\k^2/2$ edge, $\k^2/2$ aggregation and $\k^2/4$ core switches of radix $\k$, hence $\tfrac{5\k^3}{4}=5\n$ output queues in total. By Steps~2 and~3,
    \be
    \begin{aligned}
    E\sbrac{\q}&\;\ge\;\frac{1}{5\n}\cdot\frac{\m}{2}\cdot E\Big[\sum_{\ell}(Y_\ell-1)^+\Big]\\
    &\;\ge\;\frac{3}{40}\brac{\para{1-\frac p\w}^{\!\w}-(1-p)}\m\;=\;\Omega(\m).
    \end{aligned}
    \ee
    For the matching upper bound, at most $\n\m$ packets ever enter the network, so $\q$ is at most linear in $\m$.
    Hence $\q=\Theta(\m)$. Since this bound only counts aggregation-to-edge downlink queues, it also gives $\qdown=\Theta(\m)$.

\T{Step 5: bounded uplinks.} We show by induction over arrival times that no packet ever waits at an uplink. Consider a northbound packet arriving at a switch at time $t$. Each of the $\w$ southern input links of the switch carries at most one packet per period: at an edge switch, because each host sends one packet per period; at an aggregation switch, because by induction the edge uplinks do not queue, and therefore forward their packets as spaced as they receive them. Hence at most $\w-1$ other northbound packets arrived in $(t-T,t]$, and those that arrived earlier have already left, since a packet takes less than $T$ to transmit. With \simplerr, these at most $\w$ packets are consecutive northbound packets, so they are sent to distinct uplinks. With \jsq, at most $\w-1$ uplinks are busy, so the packet joins an empty one. In both cases, the packet is transmitted upon arrival, so every uplink queue holds at most one packet, and $\qup=\bigO{1}$. Conversely, $\qup=\Omega(1)$, since by \cref{eq:northbound} a constant fraction of the edge uplinks transmit one packet per period.
\end{proof}

\section{\drb Bounded Queues
    } \label{sec:DRB}

\begin{proof}[Proof of \cref{thm:DRB}] 
We want to model the queueing behavior of \hdr when using a random permutation. First, at the edge switch uplinks, we saw that we get a sum of periodic flows with random phases. However, since the permutation is random, not all of the flows are necessarily northbound. Thus, the load $\rho$ may be arbitrary. Instead, we model it using its expected value. Any of the $\w$ source hosts under an edge switch sends to another {host} internal to the same edge switch with probability $(\w-1)/(\n-1)$, therefore the average number of northbound flows is
\begin{equation}
    \w\cdot\para{1-\frac{\w-1}{\n-1}}
\end{equation}
Each such flow has a period of $\w T$, where $T$ is the packet transmission time (including inter-frame gap). We model this system as $N\cdot D/D/1${, where $N$ is this expected number of northbound flows, each of period $D=\w T$ at a given uplink port, and the load is $\rho=N/\w<1$}. We model the superposition of periodic arrivals as a Gaussian process as $\k$ increases~\cite{addie2002approximation}, together with a negative drift since $\rho<1$. Since the queue size cannot be negative, we use a truncated normal distribution. {This queue is stationary: its distribution depends on the number of flows and on $\rho$, but not on how long the flows keep sending, so the modeled average queue is independent of $\m$, \ie $\Theta(1)$.}

In the same way, for aggregation-to-core links, we neglect intra-pod traffic, and get the following expected number of inter-pod flows per source pod:
\begin{equation}
    \w^2\cdot\para{1-\frac{\w^2-1}{\n-1}}
\end{equation}
We then apply the same model. Note that the fit is not exact anymore, as traffic coming from the edge switches loses its periodicity and therefore it is important to emphasize that \textit{the ND/D/1 model is only an approximation}. Downlinks are an even coarser approximation, since each layer makes traffic less periodic.
\end{proof}

\mypar{An exact bound at the edge uplinks} For the first queueing layer, the model can be replaced by a deterministic bound. As above, each northbound flow appears at a given uplink port with period $\w T$, so it contributes at most $\frac{t-s}{\w T}+1$ packets to that port in any interval $(s,t]$. With at most $\w$ northbound flows, the port receives at most $\frac{t-s}{T}+\w$ packets in $(s,t]$ while serving one packet per $T$, so its queue never exceeds $\w+1$ packets, whatever the message size $\m$. The other layers lose this periodicity, which is why we resort to the $N\,D/D/1$ approximation there.

\begin{figure}
\centering
\begin{subfigure}[]{0.9\linewidth}
    \includegraphics[width=\linewidth]{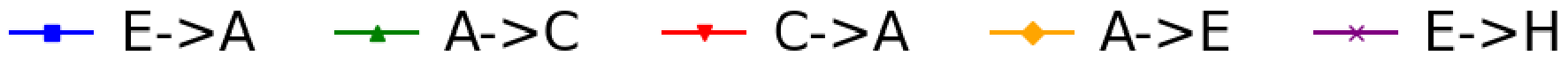}
\end{subfigure}  
\par
\begin{subfigure}{0.7\linewidth}
	\includegraphics[width=\textwidth]{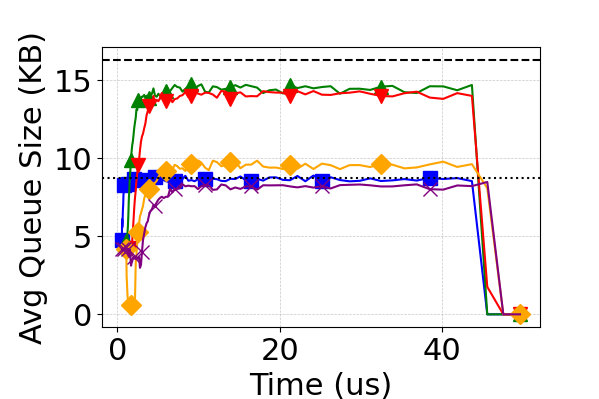}
\end{subfigure}
\caption{Queue models for \hdr. The average queue size at each of the 5 switch layers is compared against queueing models (dotted and dashed horizontal lines). }
\label{fig:drbmodel}
\end{figure}

\mypar{Validation} \Cref{fig:drbmodel} uses $\n=1,024$ hosts with a $4\unit{MB}$ message to show the behavior of \hdr at scale. It shows the average queue sizes at each of the 5 switch queueing layers: (i)~edge to aggregation ($E\to A$),   
(ii)~aggregation to core ($A\to C$),   
(iii)~core to aggregation ($C\to A$), 
(iv)~aggregation to edge ($A\to E$) and
(v)~edge to host ($E\to H$).  
It illustrates that the approximation is very close for the first E$\to$A layer. The model then is a bit pessimistic for the two core layers. It has three main sources of errors: (i)~the traffic periodicity is broken by queueing at successive switches, (ii)~the load varies and is different at each pod, and 
(iii)~we only use a {Gaussian} approximation. Finally, the two downlink edge layers are relatively close to the model as well despite the loss of periodicity.

\section{\ethereal Bounded Queues} \label{sec:ethereal}

\begin{proof}[Proof of \cref{thm:ethereal}]
We use the same model as for \hdr\ (\cref{sec:DRB}), with a single change. Under a permutation, each host sends a single flow in \drb, which \ethereal~\cite{ethereal} splits into $s$ subflows of $\m/s$ packets, one per path to the destination.
Since the NIC serves its queue pairs in round-robin~\cite{ethereal}, each subflow is a periodic source that sends one packet every $sT$, with a random phase. Hence the only change is that we model subflows, rather than flows, as sources: more sources, each with a longer period, for the same total rate.

This change does not affect the boundedness of the queue. First, each path receives exactly $\m/s$ packets of each flow, so, as with \hdr, all downlinks into a switch carry the same load and every link has load $\rho\le1$. Second, each queue is still an $N\cdot D/D/1$ queue, which is stationary, so its average does not depend on $\m$, \ie $\q=\Theta(1)$.
\end{proof}

\mypar{An exact bound at the edge uplinks} As for \hdr, the first layer admits a deterministic bound. Consider an edge uplink and a northbound flow of its edge switch. If the flow is inter-pod, $\w$ of its $\k^2/4$ subflows use that uplink, each sending one packet every $\frac{\k^2}{4}T$; if it is intra-pod, one of its $\w$ subflows does, sending one packet every $\w T$. Either way, the flow contributes at most $\frac{t-t'}{\w T}+\w$ packets to the uplink in any interval $(t',t]$. With at most $\w$ northbound flows, the uplink receives at most $\frac{t-t'}{T}+\frac{\k^2}{4}$ packets while serving one per $T$, so its queue never exceeds $\k^2/4+1$ packets, whatever $\m$. The constant is larger than \hdr's $\w+1$ because the random order of the queue pairs does not alternate between aggregation switches, so a flow may send several consecutive packets through the same uplink.

\section{Proofs and Details of \name's Properties} \label{sec:ofan}

\subsection{Correctness of Consolidation}

\mypar{Notation} Throughout this appendix, $H$ denotes a host, $E$ an edge switch, $A$ an aggregation switch and $C$ a core switch, with subscript $S$ for the source side and $D$ for the destination side. Thus $E_S$ and $A_S$ are a source edge switch and a source aggregation switch, and $H_D$, $E_D$ and $A_D$ are a destination host, its edge switch, and an aggregation switch of its pod. A \emph{downlink} is a southbound link, \ie $C\to A_D$, $A_D\to E_D$ or $E_D\to H_D$.

\mypar{Load-balancing invariant} We start by formally formulating the goal of a switch \dr in general (and \name in particular) in a network without link failures. 
It is clear that we want a perfect LB locally among uplink ports. But we also want to make sure that any source edge switch $E_S$ load-balances traffic as equally as possible across the links in the destination pods. We want to formulate a target ideal invariant that this load balancing should satisfy. However, we need to be careful in the formulation. The traffic that $E_S$ load-balances depends on the destinations of the flows that cross $E_S$. These destinations depend on the arbitrary permutation pattern, and are therefore not uniform. For example, $E_S$ cannot send the same amount of traffic to all edge switches in the network. Instead, we condition the invariant on the destination switch.

\begin{invariant}[Load-balancing invariant] \label{inv:balance}
Given any source switch {$X$} in a source pod and any destination switch {$Y$} in a destination pod, {$X$} should send traffic at the same rate through all of the southbound links to {$Y$} that are reachable from {$X$}.
\end{invariant}

The invariant implies a perfect load balancing at downlinks:

\begin{property}\label{prop:inv}
If \inv{inv:balance} holds for source aggregation switches ($X$) and destination aggregation switches ($Y$), and likewise for source edge switches and destination edge switches, then all southbound links to any switch have equal load. 
\end{property}
\begin{proof}%
    {Partition the traffic on the downlinks into a switch by its source at the layer above: the load on each core-to-$A_D$ link is the sum of the rates of the source aggregation switches that reach it, and the load on each $A_D$-to-$E_D$ link is the sum of the rates of the source edge switches that reach it. If each source sends at the same rate on all the links it reaches, these sums are equal for all the downlinks into a given switch.}
\end{proof}

\mypar{Consolidation} To achieve the above target invariant, we use \textit{consolidation into virtual flows}. The key idea is that each source switch consolidates into a single shared virtual flow all of its northbound packets %
destined to a shared  destination $Y$. Then, it load balances these packets across all of its uplink ports. However, we still need to carefully define the level at which to consolidate flows. 

\mypar{Mandatory waypoints} Consider a source aggregation switch $A_S$. A destination host $H_D$ that is reachable through its uplinks is necessarily in a different pod. Moreover, a property of the three-level fat tree is that when sending traffic from $A_S$ to $H_D$, any packet \textit{must} reach a given destination edge $E_D$ and a given destination aggregation switch $A_D$ in the destination pod, and cannot reach other edge or aggregation switches. For example, the third aggregation switch of the source pod can only send traffic destined to $H_D$ through the third aggregation switch $A_D$ of the destination pod, then through the edge switch $E_D$ that is connected to $H_D$. We can think of $A_D$ and $E_D$ as \textit{mandatory waypoints}. They stand in contrast with the many possible core switch waypoints that could be used to reach $H_D$.

As a consequence, all of the traffic going through $A_S$ and destined to $H_D$ could be consolidated into a single virtual flow using several consolidation levels: (i)~per destination host $H_D$, (ii)~per destination edge $E_D$, and (iii)~per destination aggregation $A_D$. The following result shows that they can all provide a first performance guarantee.

\begin{theorem}\label{thm:S}
Source aggregation switch $A_S$ satisfies \inv{inv:balance} whether it consolidates traffic per destination host, edge or aggregation switch.
\end{theorem}
\begin{proof}%
    Whatever the consolidation level, for any destination belonging to that level in the destination pod, $A_S$ always sends traffic at the same rate to all the cores it is connected to. Thus $A_D$, the aggregation switch of the same index in the destination pod, which is connected to those same cores, also receives the same rate from each of them, and therefore satisfies the invariant. In addition, the only downlinks $A_S$ can reach beyond $A_D$ are $A_D\to E_D$ and $E_D\to H_D$, and it can reach no other downlink into $E_D$ or into $H_D$. Being the only ones, they trivially satisfy the invariant.
\end{proof}

Consider now a source edge switch $E_S$. Its mandatory waypoints for a given destination host $H_D$ are only $H_D$ and $E_D$. Since it can reach $H_D$ using several aggregation switches, none is mandatory. Thus, $E_S$ can consolidate traffic (1)~per destination host $H_D$ and (2)~per destination edge $E_D$. We find that both work out:

\begin{theorem}\label{thm:both}
No matter the consolidation level, if both the source edge switch $E_S$ and the source aggregation switch $A_S$ consolidate traffic, then {$A_S$ satisfies \inv{inv:balance} for the downlinks into destination aggregation switches, and $E_S$ for the downlinks into destination edge switches and hosts. Hence all the downlinks into any switch have equal load.}
\end{theorem}
\begin{proof}%
    We structure the proof based on the destination type.
    
    (1) For any destination aggregation switch $A_D$, the source edge switch $E_S$ has a single source aggregation waypoint $A_S$: if the destination is in the same pod then $A_S=A_D$, and otherwise $A_S$ is the aggregation switch of the same index in the source pod. All of $E_S$'s traffic for $A_D$ is therefore sent to $A_S${. Since $A_S$ shares its pointer with the other edge switches of its pod, $E_S$'s own share need not be spread evenly over the cores; but the aggregate that $A_S$ sends satisfies the invariant for $A_D$ by \thm{thm:S}, which is all that \cref{prop:inv} needs for the core-to-$A_D$ links.}

    (2) For any destination edge switch $E_D$, the source edge switch $E_S$ splits its traffic for $E_D$ equally across all the aggregation switches of its own pod, sending the same rate $R$ to each. In turn, the rate leaving any source aggregation switch $A_S$ is carried to the aggregation switch $A_D$ of the same index in the destination pod (with $A_D=A_S$ when the two pods coincide). Hence $E_D$ receives the same rate $R$ on each of its downlinks from the aggregation switches of its pod.

    (3) Finally, destination host $H_D$ is only connected through a single link to the network, so the invariant holds trivially.
\end{proof}

Since it is clear that pointers also load-balance perfectly locally among uplink ports, the above theorem means that LB achieves its dual goals both at uplinks and at downlinks.

\subsection{Low-Cost Per-Destination Port Orders}
\label{sec:ofan-orders}

\mypar{Requirements} Ofan keeps, at each switch, one round-robin pointer per destination (per destination edge switch at an edge switch, and per destination pod at an
aggregation switch) and per packet type (\eg data or ACK). Each pointer is an index $i \in \{0,\dots,\w-1\}$ into an order of
the switch's $\w$ uplinks, starting at a random position.
Two requirements
drive the choice of order. 
First, starting from any position, a pointer must visit
every uplink exactly once per $\w$ packets, so traffic is split evenly. Second,
pointers that are active at the same time at the same switch should follow
{different} orders; otherwise they tend to select the same uplink at the same
time. Storing an independent random permutation per destination meets both
requirements, but costs $\Theta(\w \log \w)$ bits per destination.

\mypar{\name operation} Instead, \name computes the order from a single global permutation $P$ of
$\{0,\dots,\w-1\}$, shared by all switches, and a per-pointer identifier
\begin{equation}
  \mathit{id} = d + N_d \cdot t,
\end{equation}
where $d$ is the destination (edge switch or pod), $N_d$ is the number of possible
destinations, and $t \in \{0,1\}$ distinguishes data from acknowledgments. We split
this identifier into its low-order and high-order parts,
\begin{equation}
  \mathit{id}_1 = \mathit{id} \bmod 2^{b},
  \qquad
  \mathit{id}_2 = \bigl\lfloor \mathit{id} / 2^{b} \bigr\rfloor,
\end{equation}
where $b = \lceil \log_2 \w \rceil$. The $i$-th port of the order is then
\begin{equation}
  \pi(i) = \bigl( \mathrm{xor}_\w(P[i],\, \mathit{id}_1) + \mathit{id}_2 \bigr) \bmod \w,
\end{equation}
where $\mathrm{xor}_\w(v,x)$ equals $v \oplus x$ if $v \oplus x < \w$ and $v$ otherwise.

\mypar{Correctness}
Each step maps $\{0,\dots,\w-1\}$ onto itself one-to-one. Adding a constant modulo
$\w$ clearly does. For $\mathrm{xor}_\w$, since XOR is its own inverse, if $v$ is kept
because $v \oplus x \ge \w$, no other value $u < \w$ can satisfy $u \oplus x = v$,
since that would require $u = v \oplus x \ge \w$. Hence $\pi$ is a permutation, and
every pointer visits each uplink exactly once per $\w$ packets from any starting
position. Under link failures, ports that no longer lead to the destination are
skipped, so each remaining valid port still receives exactly one packet per round.

\mypar{Diversity from two non-collapsing operations}
Using a single kind of operation yields few orders: all additions combine into one
shift, giving at most $\w$ orders, and all XORs combine into one XOR, giving at most
$2^b$. XOR and addition, however, do not reduce to one another: addition preserves
arithmetic differences and breaks bitwise ones through carries, while XOR does the
opposite. Driving the XOR with $\mathit{id}_1$ (the low-order bits of $\mathit{id}$) and the addition
with $\mathit{id}_2$ (its high-order bits) therefore produces up to $2^b \cdot \w$ combinations. For
$\w = 6$ they yield 30 distinct cyclic orders, versus 6 with additions alone, so the
few destinations served concurrently by a switch almost always follow different
orders.

\mypar{Computational cost}
The per-packet work is one lookup of $P[i]$, one XOR, one comparison, one addition,
and one conditional subtraction, all on $\lceil \log_2 \w \rceil$-bit values. The
constants $\mathit{id}_1$ and $\mathit{id}_2 \bmod \w$
can be derived from the destination on the fly. 

\mypar{Storage cost} Beyond the pointers themselves, a
switch stores nothing per destination. $P$ is a single global constant of
$\w \lceil \log_2 \w \rceil$ bits that can be hard-wired and is the same for all switches. $P$ is drawn once offline
and chosen so that its $\w$ cyclic shifts define $\w$ distinct orders.

\subsection{Scaling Impact \vs \hdr}

{Each switch can consolidate at several levels (\cref{thm:S,thm:both}), and \name always chooses the highest one: the higher the consolidation level, the fewer pointers a switch holds, and the smoother its load balancing, since each pointer coordinates more flows. We now quantify this gain \vs \hdr.}

\mypar{Scaling impact on pointer numbers} Consider a random permutation with inter-pod traffic only. In \name, aggregation switches consolidate traffic per destination aggregation switch. Thus, a destination aggregation switch can only receive up to $\k-1$ different virtual flows, since only the $\k-1$ source aggregation switches of the same index, one per other pod, can reach it. However, in \hdr, they  virtually consolidate traffic per destination host, hence there are $\k^2/4$ flows, corresponding to the $\k^2/4$ hosts in the destination pod. Thus, the consolidation gain in the number of pointers is at least $\k^2/(4(\k-1))>\k/4,$ increasing with the network size. 

\begin{figure}
    \centering
    \includegraphics[width=0.55\linewidth]{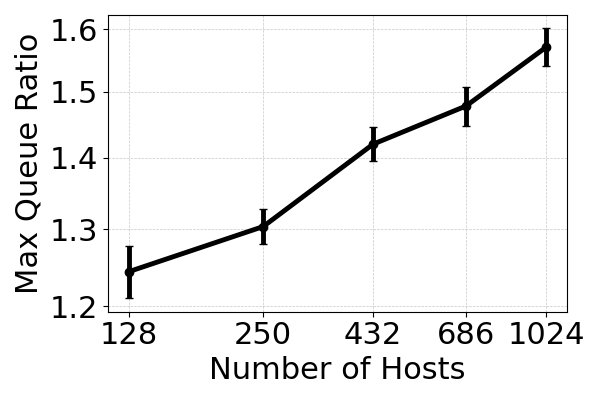}
    \caption{Ratio of the maximum queue size of \drb to that of \name as a function of the number of hosts. 
    }
    \label{fig:ratio}
    \vspace{-1em}
\end{figure}

\mypar{Scaling impact on queue sizes} {Under the same traffic,} \Cref{fig:ratio} plots the ratio of the maximum queue size of \drb to that of \name as a function of network size, with 30 runs per size. The ratio increases with the network size, \ie \name outperforms even more as we scale the network. 

\subsection{Model with Bounded Average Queue Sizes }

{We now build a model for the queue size under \name.}
\begin{proof}[Proof of \cref{thm:ofan}] 
The model follows similar lines for \name as for \hdr. However, since we consolidate pointers, we   evaluate the number of distinct pointers, and model as if all pointers carried the same number of periodic flows, ignoring the variability across pointers.

We start with the edge model. We want to determine the expected number of pointers at the edge switch. A single source edge switch comprises $\w$ flows (balls), and these flows can be destined to $\k^2/2$ edge switches (bins). The expected number of distinct destination edge switches is:
$$ \frac{\k^2}{2}\cdot \para{ 1 - \left( 1 - \frac{1}{\k^2/2} \right)^{\w} }. $$
Using Taylor expansion, we find that this is about $\w-1/4$.

Likewise, at each aggregation switch, the number of northbound pointers is at most $\k-1$, as there are $\k-1$ other pods and therefore other aggregation switches that can be destinations. The expected number of pointers also converges to $\k-1$ as $\k$ grows. Thus, we obtain a new model, which neglects many aspects of the algorithm that can be too complex to capture, and in particular the variability in the periods of each virtual flow. Again, \textit{the ND/D/1 model is only an approximation}. {As for \hdr, the modeled queue is stationary and therefore independent of $\m$.}
\end{proof}

{\mypar{An exact bound at the edge uplinks} As for \hdr (\cref{sec:DRB}), the first queueing layer admits a deterministic bound. A pointer shared by $f$ flows advances $f$ times per period $T$, and therefore visits each uplink port exactly $f$ times per $\w$ periods, \ie it contributes at most $f\para{\frac{t-s}{\w T}+1}$ packets to a given port in any interval $(s,t]$. Since the (at most $\w$) northbound flows of an edge switch are split among its pointers, the port receives at most $\frac{t-s}{T}+\w$ packets in $(s,t]$ while serving one packet per $T$, so its queue never exceeds $\w+1$ packets, whatever the message size $\m$.}

\begin{figure}
\centering
\begin{subfigure}[]{0.9\linewidth}
    \includegraphics[width=\linewidth]{fig/lmodel.png}
\end{subfigure}  
\par
\begin{subfigure}[]{0.7\linewidth}
	\includegraphics[width=\textwidth]{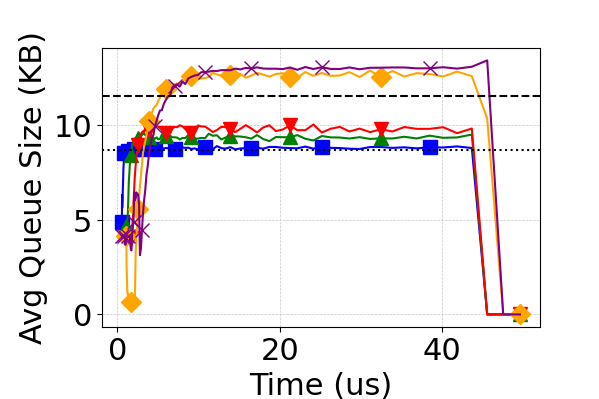}  
\end{subfigure}
\caption{Queue models for \name. The average queue size at each of the 5 switch layers is compared against queueing models (dotted and dashed horizontal lines). }
\label{fig:ofanmodel}
\end{figure}

{\mypar{Validation}} \Cref{fig:ofanmodel} illustrates the model for \name under the same conditions. We can see that the edge model holds well (dotted line). However, as expected, the models beyond have a slightly poorer fit. In particular, the downlink queues are  harder to model exactly.

\subsection{Dynamic First Pointer}\label{sec:advanced}

In a network without link failures, we find that if (i)~the first pointer is not picked at random, but rather chooses the shortest queue at the arrival of the first packet, then (ii)~it cycles in an unshuffled round-robin order, then CCT is slightly improved. However, 
switches would need to know the start of a collective, or restart pointers after a timeout; we do not further evaluate this variant in this paper.

\section{Proof of \cref{thm:pkt}: Optimal Packet Size} \label{sec:pkt}

\begin{proof}[Proof of \cref{thm:pkt}] 
The CCT can be modeled as the sum of (i)~the time to send the message, (ii)~the time to go through the queues, and (iii)~the zero-load RTT. Only the first two components depend on the packet size $P$, so let's focus on them. 

(i)~The time to send the message is
\be \frac{D}{P-H}\cdot \frac{P}{\C},
\ee
since we need $\frac{D}{P-H}$ packets to send a message of $D$ bytes using $P-H$ payload per packet, and each packet of size $P$ that is sent at rate $\C$ takes time $\frac{P}{\C}$. 

(ii)~The time to go through the queues is
\be \alpha \cdot \frac{P}{\C},
\ee
by assumption, since the queues hold $\alpha$ packets of size $P$.

Adding the two components and dropping the common factor $\frac1\C$, we want to minimize
\be
f(P):=P\cdot \para{\frac{D}{P-H}+\alpha}=\frac{PD}{P-H}+\alpha P .
\ee
Differentiating with respect to $P$, 
we get
\be
f'(P)=-\frac{D\cdot H}{\para{P-H}^2}+\alpha .
\ee
Zeroing the derivative gives $\para{P-H}^2=\frac{DH}{\alpha}$, that is $P-H=\sqrt{\frac H\alpha D}$, which is the claimed result. It is indeed a minimum, since $f''(P)=\frac{2DH}{\para{P-H}^3}>0$ for every $P>H$. 
\end{proof}

Note that \cref{fig:pktsize} assumes as a first approximation $H=82\unit{B}$, accounting for the header and inter-frame gap components, and $\alpha=10$ data packets.

\section{Proof of \pipe Queueing Model (\cref{thm:pipe})} \label{sec:pipe}

\begin{proof}[Proof of \cref{thm:pipe}]
We use the same model as in the proof of \cref{thm:ofan} (\cref{sec:ofan}), with a single change. In \name, each key ((i)~destination edge switch at an edge switch or destination pod at an aggregation switch, and (ii)~size class) has a single pointer per switch; in \pipe, each key has one independent pointer per ingress pipe, and a packet uses the pointer of the pipe it enters through. Each of these pointers still rotates over all the switch's uplinks. Hence the only change is that more pointers, \ie more periodic sources, feed each queue: up to $P$ times more for $P$ ingress pipes.

This change does not affect the boundedness of the queue, for two reasons. First, the load on every link stays balanced: each pipe pointer spreads its own traffic equally over its uplinks, so the sum of the pipe pointers of a key does too, exactly like \name's single pointer. The proofs of \cref{thm:S,thm:both} therefore hold unchanged, and every link keeps a load $\rho\le1$ by \cref{prop:inv}. Second, the modeled $N\cdot D/D/1$ queue is stationary for any number $N$ of sources: a larger $N$ changes the value of the average queue, but not the fact that it is independent of $\m$. Since $N$ is at most $P$ times \name's number of pointers, which does not depend on $\m$, we get $\q=\Theta(1)$.
\end{proof}

\mypar{An exact bound at the edge uplinks} The deterministic bound of \cref{sec:ofan} also carries over unchanged, since it only used that the (at most $\w$) northbound flows of an edge switch are split among pointers such that a pointer shared by $f$ of these flows visits every uplink port $f$ times per $\w$ periods, whatever the split. Hence the queue of any edge uplink under \pipe is at most $\w+1$ packets, whatever the message size $\m$.

\mypar{The cost of pipes} The price of the pipes is a larger number of pointers, hence a larger constant. At an edge switch the difference is small, since \name's $\approx\w-1/4$ edge pointers mostly carry a single flow each already. At an aggregation switch with $P$ ingress pipes, however, \name's at most $\k-1$ pointers become up to $P(\k-1)$, \eg $P=\k/(2W)=8$ for $\k=64$ and pipes of $W=4$ ports. This is consistent with the gap between \pipe and \name widening with the network size in \cref{fig:vsDRB}, while both remain bounded in $\m$. As for \name, the model assumes that the pointers are randomly initialized so that they do not lock in phase, and it is only an approximation beyond the edge layer.

\section{\name's Hardware Implementation}
\label{sec:implementation-details}

\mypar{Concurrent access challenge} Implementing \name's per-destination-group pointers at line rate raises a structural challenge: a switch processes many packets per cycle in parallel, arriving from different input ports, and several of these packets can be destined to the same destination edge or aggregation switch within the same cycle. Such packets must all access the \emph{same} \name round-robin (RR) pointer, creating a read-modify-write hazard on a single shared register.

\mypar{Combining network} This concurrent-access problem is not new: it is a classic problem in parallel computing, where many concurrent processors issuing fetch-and-add requests to a shared counter can create the same kind of contention~\cite{kruskal1988efficient,lebeck1994request}. The classical fix, called a \emph{combining network}, is to merge duplicate concurrent requests together before they reach the shared counter, so the counter is only touched once no matter how many requests asked for it.

That said, these classical papers solve a more general problem, where requests to a huge, dynamically-addressed memory are routed hop-by-hop through a large network (\eg an omega or butterfly network) and only combine when they happen to collide en route. This general machinery is costly and has variable latency, which does not fit a switch pipeline that needs guaranteed, fixed per-cycle latency. Instead, we want to solve our narrower problem, over a small fixed set of $N$ local pointers, using a much cheaper, fixed-latency, purpose-built circuit, described next.

\mypar{\name pipeline} We can implement \name's pointers with a simple, fixed-latency pipeline that processes a batch of $\w$ packets in parallel, where $\w = \k/2$ is the number of downlink ports (only packets arriving from downlinks need a pointer decision, since uplink paths are chosen while downlink paths are unique). The pipeline has four stages: (1) sort the batch by pointer key, so packets sharing a pointer become adjacent; (2) combine the counts of packets sharing each pointer key; (3) update the shared pointer register for each key; and (4) decombine the updated pointer values back to each packet in the batch.
We illustrate the pipeline in our fat-tree example of \cref{sec:algo} with $\k=64$-port switches, \ie at most $\w=\k/2=32$ packets arriving in parallel from downstream inputs, $32$ uplink ports to cycle through, and $N=2K$ edge switch destinations.

\Ts{(1) Sort.} First, a $\w$-input bitonic sorting network~\cite{chen2017computer} sorts the batch by pointer key, so packets sharing a pointer become adjacent.

\Ts{(2) Combine.} A combining tree~\cite{dimitrakopoulos2021sum} computes, in $\log_2{(\w)}$ stages, how many packets in the batch share each key; the last packet of each group becomes its \emph{master}, carrying the group's match count.

\Ts{(3) Update.} A register file of $N$ per-pointer counters is updated: each master atomically fetches its pointer's current value and advances it, modulo the number of uplink ports, by the group's match count.

\Ts{(4) Decombine.} Finally, a reverse network decombines the fetched value back to every packet in its group (offsetting it by the packet's position within the group), and an unsorting network---the inverse of the initial permutation---restores the original per-lane order, so each packet leaves the pipeline carrying its own assigned uplink port. For our running example, this requires only a $32$-input sorting/combining tree and $2K$ registers per edge switch, a practical hardware budget.

\mypar{Example} Suppose 3 of the $\w$ packets in a batch are all destined to the same edge switch, which has $\k/2=32$ uplink ports, so its pointer takes values in $\{0,1,\ldots,31\}$ and currently holds the value $1$. In stage (1), sorting by pointer key brings these 3 packets together, adjacent to each other. In stage (2), the combining tree counts the 3 matches and marks the last of them as the \emph{master}, carrying a match count of $3$. In stage (3), this master atomically fetches the pointer's current value ($1$) and advances it by $3 \bmod 32$, leaving the pointer at $(1+3) \bmod 32 = 4$ for the next batch. In stage (4), the 3 packets are decombined: each is assigned the fetched value offset by its position within the group, \ie ports $1$, $2$, and $3$ respectively, and the unsorting network routes each result back to the packet's original lane. Thus, the pointer is touched exactly once for all 3 packets, yet each packet still receives its own distinct, correctly-ordered round-robin port, exactly as if the 3 requests had been processed one at a time.

\mypar{Cost estimate} The pipeline requires $N$ registers for the pointers, plus a $\w$-input sorting network, a $\w$-input combining tree, and a matching decombining/unsorting network of the same size. 
For our running example, recall each edge switch keeps one pointer per (destination edge switch, packet-size quantization group) pair; with $2K$ destination edge switches and $2$ quantization groups (data, ACKs), $N \approx 4K$. Each pointer needs $\log_2(32)=5$ bits to represent its value (one of the $32$ uplink ports), so the register file is only $N \times 5 \approx 20K$ bits ($\approx 2.5$KB). The rest of the pipeline is dominated by two contributions: (i) the $W{=}32$-input sorting and unsorting networks, each with $15$ stages of $16$ compare-and-swap cells (roughly $80$ gates per cell, for a $12$-bit pointer-key comparator, wide enough to address $N\approx4K$ keys, plus payload-swap muxes), for $\approx 20K$ gates apiece; and (ii) the register file's read-modify-write logic, dominated by the $20K$ flip-flops themselves at $\approx 4$ gate-equivalents each, \ie $\approx 82K$ gates. The combining tree adds only a couple thousand more gates ($\approx 80$ small adder nodes for the segmented match counts). Altogether, this is on the order of $10^5$ logic gates ($\approx 127K$). At a 1~GHz clock and a conservative 10\% switching activity, this amounts to only a few milliwatts of dynamic power ($\approx 6$mW). 

Both figures are negligible next to the resources of the very chip we used as our pipe-architecture example: the BCM56980 (Tomahawk~3) datasheet lists a typical power draw of $300$~W and $64$~MB of integrated packet-buffer memory~\cite{broadcom-bcm56980}, independently corroborated by~\cite{wheeler2019tomahawk4}. So \name's pipeline's $\approx 6$mW would add only $\approx 0.002\%$ to such a chip's power budget, and its $\approx 20K$-bit register file is only $\approx 0.0037\%$ of that chip's $64$MB ($\approx 537$ million bits) of buffer memory---both entirely negligible.

\mypar{Other details} Implementing \name only requires modifying the core scheduling algorithm in a switch. No change is required at endhosts and \name can continue to use existing routing protocols such as BGP W-ECMP~\cite{ietf-bess-ebgp-dmz-08,ietf-idr-link-bandwidth-24,rfc7938, nokia-w-ecmp,nvidia-w-ecmp, juniper-w-ecmp,arista-w-ecmp,huawei-w-ecmp} to propagate information about failures \etc

\end{document}